\documentclass[11pt,letterpaper]{article}

\usepackage[utf8]{inputenc}
\usepackage{amsmath,amsthm,amssymb}
\usepackage{fullpage}
\usepackage[ruled]{algorithm2e}
\usepackage{tikz}
\usepackage{esvect}
\usepackage{caption}
\usepackage{subcaption}
\usepackage[left= 3 cm,right=3 cm,top=3 cm,bottom=3 cm]{geometry}
\usepackage{mathtools}
\usepackage[colorlinks = true, citecolor = {blue}]{hyperref}
\usepackage{array}
\usepackage{comment}
\usepackage{authblk}
\usepackage{thm-restate}
\usetikzlibrary{positioning}
\usetikzlibrary{patterns}
\usetikzlibrary{decorations}
\tikzset{
bicolor/.style 2 args={
  dashed,dash pattern=on 10pt off 10pt,#1,
  postaction={draw,dashed,dash pattern=on 10pt off 10pt,#2,dash phase=10pt}
  },
}

\newtheorem{theorem}{Theorem}
\newtheorem{lemma}{Lemma}
\newtheorem{corollary}{Corollary}
\newtheorem{definition}{Definition}

\newcommand{\defoptproblem}[3]{
 \vspace{1mm}
\noindent\fbox{
 \begin{minipage}{0.96\textwidth}
 #1 \\
 {\bf{Input:}} #2 \\
 {\bf{Output:}} #3
 \end{minipage}
 }
 \vspace{1mm}
}

\newcommand{\set}[1]{\left\{ #1 \right\}}

\newcommand{\ifend}{\textbf{endif}}

\newcommand{\diam}{\mbox{\textsf{diam}}}

\newcommand{\wecc}[2]{\mbox{\textsf{ecc}}\left(#1 ~\vert ~ #2\right)}

\newcommand{\ulog}{\mathcal{U}_{2\log}}
\newcommand{\ut}{\mathcal{U}_{3}}
\newcommand{\med}{\mbox{Med}}
\newcommand{\lad}{{\tt lad}}
\newcommand{\umax}{u_{\max}}
\newcommand{\algo}{\textsc{M$\Theta$rse}}
\newcommand{\stz}{\mbox{St}}

\newcommand{\lacla}{\Lambda_G^{\Theta}}
\newcommand{\laside}{\Lambda_G^{\mbox{\scriptsize{side}}}}
\newcommand{\ladist}{\Lambda_G^{\mbox{\scriptsize{dist}}}}
\newcommand{\ladiste}{\Lambda_G^{\emph{\scriptsize{dist}}}}
\newcommand{\lagate}{\Lambda_G^{\mbox{\scriptsize{gate}}}}
\newcommand{\larec}{\Lambda_G^{\mbox{\scriptsize{rec}}}}
\newcommand{\larece}{\Lambda_G^{\emph{\scriptsize{rec}}}}
\newcommand{\lafib}{\Lambda_G^{\mbox{\scriptsize{fib}}}}
\newcommand{\lafibe}{\Lambda_G^{\emph{\scriptsize{fib}}}}
\newcommand{\laarg}[1]{\Lambda_G^{#1}}

\title{Quasilinear-time eccentricities computation, and more, on median graphs}

\author[1]{Pierre Berg\'e}
\author[2,3]{Guillaume Ducoffe}
\author[4]{Michel Habib}
\affil[1]{Université Clermont-Auvergne, CNRS, Mines de Saint-Etienne,
	Clermont-Auvergne-INP, LIMOS, 63000 Clermont-Ferrand, France}
\affil[2]{National Institute for Research and Development in Informatics, Romania}
\affil[3]{University of Bucharest, Romania}
\affil[4]{IRIF, CNRS \& Universit\'e Paris Cité, France}

\date{}

\begin{document}

\maketitle

\begin{abstract}
Computing the diameter, and more generally, all eccentricities of an undirected  graph is an important problem in algorithmic graph theory and the challenge is to identify graph classes for which their computation can be achieved in subquadratic time.
Using a new recursive scheme based on the structural properties of median graphs, we provide a quasilinear-time algorithm to determine all eccentricities for this well-known family of graphs.
Our recursive technique manages specifically balanced and unbalanced parts of the $\Theta$-class decomposition of median graphs. The exact running time of our algorithm is $O(n\log^4 n)$. This outcome not only answers a question asked by B\'en\'eteau {\em et al.} (2020) but also greatly improves a recent result which presents a combinatorial algorithm running in time $O(n^{1.6408}\log^{O(1)} n)$ for the same problem. 

Furthermore we also propose a distance oracle for median graphs with both poly-logarithmic size and query time. Speaking formally, we provide a combinatorial algorithm which computes for any median graph $G$, in quasilinear time $O(n\log^4(n))$, vertex-labels of size $O(\log^3(n))$ such that any distance of $G$ can be retrieved in time $O(\log^4(n))$ thanks to these labels. 
\end{abstract}

\bigskip

\section{Presentation of the contributions: ideas and impact} \label{sec:intro}

A wide literature in graph theory, and more specifically in metric graph theory, is dedicated to median graphs. Given two vertices $u,v \in V(G)$, the interval $I(u,v)$ is the set containing all vertices $x$ metrically between $u$ and $v$, {\em i.e.} $x\in I(u,v)$ if $d(u,v) = d(u,x) + d(x,v)$. 
Formally, median graphs are the graphs $G$ such that for any triplet of distinct vertices $u,v,w \in V(G)$, the intersection $I(u,v) \cap I(v,w) \cap I(w,u)$ is a singleton. From the applications point of view, median graphs are crucial in the study of phylogenetic networks~\cite{BaFoSyRi95,PaBe12}. They are usually mentioned as ``median networks'' in this area of research. They are, in particular, widely considered in the visualization of sequence variations in human mitochondrial DNA~\cite{BaFoRo99,BaMaRi00,BaQuSaMa02,Zi11}.
From a more theoretical point of view, median graphs are in bijection with numerous and diverse notions from discrete mathematics. For example, they represent solutions of 2-SAT formulae~\cite{MuSc79,Sc78}. In geometric group theory, median graphs are exactly the 1-skeletons of CAT(0) cube complexes~\cite{BaCh08,Ch00}. In abstract models of concurrency, median graphs are in bijection with event structures~\cite{BaCo93,SaNiWi93}. Finally, median graphs form a natural subclass of \emph{partial cubes}, {\em i.e.} isometric subgraphs of hypercubes~\cite{Mu78}.

From a graph theory point of view now, median graphs admit many characterizations and structural properties. First, median graphs are bipartite (hence, triangle-free) and do not admit any induced $K_{2,3}$. They are exactly the retracts of hypercubes~\cite{Ba84}, and the Cartesian product of two median graphs is also median~\cite{BaCh08}. Median graphs are sparse: the number $m$ of edges of any median graph $G$ satisfies $m \le n\log_2 n$, where $n = \vert V(G) \vert$. Among the most famous subclasses of median graphs, we can cite: trees, hypercubes, grids (of any dimension), and cogwheels. A very practical characterization of median graphs come from the notion of $\Theta$-classes which provide us with natural edge separators satisfying convexity properties. This concept, used in most of the algorithmic literature cited below, will be a key instrument of our contributions in this article.

Recently, several efficient algorithms, dedicated to solving distance problems on median graphs, have been proposed~\cite{BeChChVa20,BeDuHa22,BeHa21,Ch12,ChLaRa19,GeCoMiNa22}. An important outcome is the algorithm proposed by B\'en\'eteau {\em et al.}~\cite{BeChChVa20} which computes both the median set and the Wiener index of median graphs $G$ in linear time $O(m)$, where $m = \vert E(G) \vert$. A second one~\cite{BeDuHa22} computes all eccentricities of median graphs in subquadratic time $\tilde{O}(n^{1.6408})$, where the $\tilde{O}$ notation neglects poly-logarithmic factors. Constant $1.6408$ comes from a slight improvement on a first version of the algorithm, which achieved $\tilde{O}(n^{\frac{5}{3}})$.

Other contributions focused on subclasses of median graphs too. There is a linear-time algorithm~\cite{BeHa21} which determines the diameter of median graphs for which the dimension of the largest induced hypercube is bounded. Furthermore, Chepoi {\em et al.}~\cite{ChLaRa19} designed distance and routing labeling schemes of $O(\log^3 n)$ bits for cube-free median graphs.

\textbf{Our contributions}. The main outcome of this paper consists in the proposal of a combinatorial algorithm which computes all eccentricites of a median graph in quasilinear time. Before stating formally this contribution, let us recall the problem we treat here. The distance $d(u,v)$ between two vertices $u,v \in V(G)$ is the length of a shortest $(u,v)$-path. Given a vertex $u \in V(G)$, its \textit{eccentricity} $\wecc{u}{G}$ is the maximum distance from $u$ to any other vertex of $G$. The \textit{diameter} and \textit{radius}, which are certainly the most studied metric parameters on graphs, correspond respectively to the maximum/minimum eccentricity of the graph.

\defoptproblem{\textsc{Eccentricities}}{A median graph $G = (V,E)$.}{All labels $\wecc{u}{G} = \max \{d(u,v) : v \in V\}$ for each vertex $u \in V$.}

In this article, we focus in fact on a more general problem, which is a weighted version of the \textsc{Eccentricities} problem. The input graph is a vertex-weighted median graph and the objective is to determine, for each vertex $u$, its \textit{weighted eccentricity} $\wecc{u}{(G,\omega)}$, which is the maximum value $d(u,v) + \omega(v)$. In other words, the weight of the arrival vertex is added to the standard distance. Obviously, it generalizes \textsc{Eccentricities} since fixing weights $0$ to each vertex is equivalent to the classical problem.

\defoptproblem{\textsc{Weighted Eccentricities}}{A weighted median graph $G = (V,E,\omega)$, with $\omega : V \rightarrow \mathbb{N}$.}{All labels $\wecc{u}{(G,\omega)} = \max \{d(u,v) + \omega(v) : v \in V\}$ for each $u \in V$.}

The main contribution of this article is thus stated below. Observe that it greatly enhances the literature as, for now, the best running time for computing all eccentricies of a median graph was in $\tilde{O}(n^{1.6408})$, from~\cite{BeDuHa22}. Moreover, we handle a more general problem.

\begin{restatable}{theorem}{mainEcc}
	\label{th:main_ecc}
	There exists a combinatorial algorithm which computes all weighted eccentricities of a weighted median graph $(G,\omega)$ in quasilinear time $O(n\log^4(n))$.
\end{restatable}

This result is described in Sections~\ref{sec:first_step} and~\ref{sec:unbalanced}. The idea of our algorithm is in fact relatively simple: $\Theta$-classes, which are presented in the preliminary Section~\ref{sec:notation}, are separators of the input graph $G$ and admit several convexity properties. Hence, inspired by well-known divide-and-conquer methods on graphs~\cite{LiTa79}, we make a tradeoff on the balance of these separators: either the sizes of the two sides generated by the separator are comparable, or not. Thus, we distinguish two cases. 

When the median graph $G$ admits a balanced $\Theta$-class (this notion is defined in Section~\ref{subsec:balanced}), we retrieve all weighted eccentricities by recursively focusing on each side of the separator, with an extra procedure that runs in linear time. As the $\Theta$-class is balanced, there is a non-negligible decrease of the size of each side compared to $n = \vert V(G) \vert$. After obtaining the weighted eccentricities of each side recursively, our extra procedure consists in retrieving the weighted eccentricities of the whole graph by looking for large distances between vertices of different sides. This can be achieved thanks to the \textit{gatedness} of each side: this property of $\Theta$-classes will be recalled in Section~\ref{subsec:def}.

However, when no $\Theta$-class of the input graph is balanced, the previous technique is not efficient anymore since for any $\Theta$-class, the size of its sides does not decrease enough. We provide a list of characterizations of median graphs without balanced $\Theta$-classes. In particular, we observe that such graphs admit a unique median vertex $v_0$. Then, we prove how a BFS starting from $v_0$ together with a reasonable number of recursive calls on convex subgraphs help us in finding all weighted eccentricities of $(G,\omega)$. This second part of our algorithm is more technical than the first one.

As a second result, we propose a \textit{distance oracle} (DO) for median graphs. The objective beyond our DO is to locally store some information
so that one can retrieve the distance between any pair of vertices by inspecting the labels in a very short time. Concretely, we propose an algorithm which assigns a label to each vertex of the graph, and then we show how these labels can be used to compute fast any value $d(u,v)$. Observe however that our labeling is not a \textit{distance labeling scheme}, since we might need to look at more than two vertex labels in order to compute some distance $d(u,v)$. Indeed, the information given by the labels of $u$ and $v$ might not be sufficient to obtain $d(u,v)$, extra labels must be taken into account. Our result is described in Section~\ref{sec:do}.

\begin{restatable}{theorem}{mainDO}
	\label{th:main_do}
	There exists a combinatorial algorithm which computes in quasilinear time $O(n\log^4(n))$, for any median graph $G$, vertex-labels $(\Lambda_G(u))_{u \in V(G)}$ of size $O(\log^3(n))$ such that the distance $d(x,y)$ between a pair $x,y$ of vertices can be retrieved in time $O(\log^4(n))$ thanks to the labels.
\end{restatable}

The proof re-uses the splitting between the balanced and unbalanced cases. When a $\Theta$-class is balanced, we apply similar arguments to those used for the computation of eccentricities, with the difference that the necessary information is put into vertex labels. Concretely, any vertex $u$ is labeled with the identity of the balanced $\Theta$-class considered, the side of $u$, the closest-to-$u$ vertex on the other side (its \textit{gate}), the distance to its gate, and finally the label of $u$ in the graph induced by its side regarding the balanced $\Theta$-class. This whole package allows us to retrieve any distance between two vertices in poly-logarithmic time.

The unbalanced case requires more effort. Our idea consists in partitioning all vertices of the graph in function of their ``direction'' regarding the central vertex $v_0$. To do so, we launch a BFS from $v_0$ and take note of some information, that will be added to the label of any vertex $u\neq v_0$: the distance from $v_0$ to $u$, the $\Theta$-classes traversed to go from $v_0$ to $u$, etc. In addition, some gates of $u$ through certain small convex sets are also computed. With this information, we show again how to retrieve any distance in poly-logarithmic time.


\textbf{Perspectives}. First, we believe that the techniques proposed in this article offer not only tools for different problems on median graphs but also for more general classes of graphs. Observe that larger families of graphs are still impacted by the notion of $\Theta$-class, the main difference consisting in weaker convex characterizations: almost median graphs~\cite{Br07}, pseudo-median graphs~\cite{BaMu91,Ve02} or partial cubes~\cite{Wi84}. In our work, we exploit several times the fact that the \textit{boundary} of each $\Theta$-class is convex, which is a property specific to median graphs and not to these superclasses. Therefore, one should be able to get rid of this argument in order to handle larger families of graphs. However, a tradeoff on the balance of $\Theta$-class stays, in our opinion, a promising starting point for tackling them.

Coming back to median graphs, one can hope producing efficient algorithms by exploiting again this tradeoff technique. A future direction of research could be trying to design algorithms which improve the naive general method for computing other metric parameters, such as the \textit{hyperbolicity}~\cite{Gr87}, the \textit{betweenness} and \textit{reach centralities}~\cite{AbGrWi15},\dots The techniques proposed in our paper can be useful tools for such problems. Eventually, we mention a problem which was not studied yet on median graphs to the best of our knowledge: 
\textsc{Weighted Center}~\cite{BeBhShTa07}. Given a median graph $G$ with vertex weights $\omega : V \rightarrow \mathbb{N}$, the objective is to determine the \textit{weighted center} of $(G,\omega)$, {\em i.e.} the vertex $u$ which minimizes $\max_{v\in V(G)} \omega(v)d(u,v)$. Observe that, with Theorem~\ref{th:main_ecc}, we can deduce, from the weighted eccentricities, some kind of weighted center where weights stand as an additive term and not multiplicative. For this reason, we believe that the techniques we proposed can be fruitful for solving \textsc{Weighted Center} in quasilinear time. Note that \textsc{Weighted Center} admits exact quasilinear-time algorithms for trees and cactii~\cite{BeBhShTa07}, hence targeting median graphs is a natural challenge.

\section{Preliminaries} \label{sec:notation}

We begin with a reminder of some notions of graphs, and more particularly on median graphs. We emphasize on a very important tool in this area: $\Theta$-classes, which are equivalences classes over the edge set of median graphs, and especially the orthogonality between these $\Theta$-classes. From the mathematical point of view, notation $\log$ refers to the natural logarithm, {\em i.e.} $\log n = \log_e n$. When another base is considered, we mention it as a subscript, {\em e.g.} $\log_2$.

\subsection{Definitions and properties from metric graph theory} \label{subsec:def}

All graphs $G = (V,E)$ considered in this paper are undirected, simple (loopless and without multiple edges), finite and connected. To avoid confusions when several graphs are considered, we denote by $V(G)$ (resp. $E(G)$) the vertex set (resp. edge set) of $G$. Usually, we use $n$ to denote the size of the vertex set of $G$, {\em i.e.} $\vert V(G) \vert$, while $m$ denotes the size of the edge set: $m = \vert E(G) \vert$. To improve readibility, edges $(u,v) \in E$ are sometimes denoted by $uv$. Let $N(u)$ be the \textit{open neighborhood} of $u \in V$, {\em i.e.} the set of vertices adjacent to $u$ in $G$. We extend it naturally: for any set $A \subseteq V$, the neighborhood $N(A)$ of $A$ is the set of vertices outside $A$ adjacent to some $u \in A$.

A subgraph $G'$ of $G$ is a graph $G'=(V',E')$, where $V' \subseteq V$ and $E' \subseteq E\cap (V' \times V')$. 
For any $U \subseteq V$, let $E\left[U\right]$ be the set of edges of $G$ with two endpoints in $U$. We denote by $G\left[U\right]$ the subgraph of $G$ induced by $U$: $G\left[U\right] = \left(U,E\left[U\right]\right)$. 

Given two vertices $u,v \in V$, let $d(u,v)$ be the \textit{distance} between $u$ and $v$, {\em i.e.} the length of a shortest $(u,v)$-path.  When the context is not clear enough, we might notice, as a subscript, in which graph these distance are considered, {\em e.g.} $d_G(u,v)$. As a weighted version is defined in the remainder, we may refer to this notion as the \textit{unweighted distance}. The \textit{eccentricity} $\wecc{u}{G}$ of a vertex $u \in V$ in graph $G$ is the length of a longest shortest path starting from $u$. Put formally, $\wecc{u}{G}$ is the maximum value $d(u,v)$ for all $v \in V$: $\wecc{u}{G} = \max_{v \in V} d(u,v)$.
The diameter of graph $G$ is the maximum distance between two of its vertices: $\diam(G) = \max_{u \in V} \wecc{u}{G}$. 

We denote by $I(u,v)$ the \textit{interval} of pair $u,v \in V$. It contains exactly the vertices which are metrically between $u$ and $v$:
$I(u,v) = \set{x \in V: d(u,x) + d(x,v) = d(u,v)}$. The vertices of $I(u,v)$ are lying on at least one shortest $(u,v)$-path.

\begin{definition}[Convex and gated sets]
We say that a set $H\subseteq V$ (or equivalently the induced subgraph $G\left[H\right]$) is \emph{convex} if $I(u,v) \subseteq H$ for any pair $u,v \in H$. Moreover, we say that $H$ is \emph{gated} if any vertex $v \notin H$ admits a $H$-\emph{gate} $g_H(v) \in H$, {\em i.e.} a unique vertex that belongs to all intervals $I(v,x)$, $x\in H$. For any $x \in H$, we have $d(v,g_H(v)) + d(g_H(v),x) = d(v,x)$. 
\end{definition}

Observe that, if $v \in H$, then it admits a natural $H$-gate: itself, since $v$ belongs not only to $H$ but also to all intervals $I(v,x)$, $x \in H$. Gated sets are convex by definition. Indeed, by contradiction, for a gated set $H$, if a shortest path from $x\in H$ to $y\in H$ was containing a section outside $H$, it would imply the existence of two $H$-gates for the vertices of this section. Conversely, convex sets are not necessarily gated in general ({\em e.g.} any pair of vertices in the $3$-clique $K_3$). But, we will see in the remainder that, on median graphs, convexity and gatedness are equivalent notions. 

A well-known property is that the intersection of two gated sets is itself gated.

\begin{lemma}[Intersection of gated subgraphs~\cite{BaCh08}]
Given two gated sets $H_1,H_2$ of a graph $G$, the set $H_1\cap H_2$ is gated. 
\label{le:intersection}
\end{lemma}

We naturally focus on the set of vertices which admit a given vertex as a gate.

\begin{definition}[Fibers~\cite{BaCh08}]
Given a gated set $H$ of $G$ and a vertex $x \in H$, the \textit{fiber} $F_H[x]$ is the set of vertices which admit $x \in H$ as a gate for $H$. \end{definition}

As each vertex in $H$ is its own gate, the fibers $F_H[x]$, $x \in H$, partition  $V(G)$. A fiber is thus a set of vertices which all admit the same gate for a given gated $H$. In fact, the $H$-gate of some $v \notin H$ is necessarily the vertex of $H$ minimizing the distance to $v$. Observe that, except $x$ itself, $F_H[x]$ contains vertices outside $H$: $F_H[x] \setminus \{x\} \subseteq V(G) \setminus H$. In the following sections, we often manipulate this set $F_H[x] \setminus \{x\}$ that we call the \emph{open fiber}.

\begin{definition}[Open fibers]
Given a gated set $H$ of $G$ and a vertex $x \in H$, we define the \emph{open fiber} as $F_H(x) = F_H[x] \setminus \{x\}$.
\label{def:open_fiber}
\end{definition}

There is a natural linear-time algorithm (in fact, a very slight variation of BFS) for computing the fibers (and open fibers) of a given gated set $H$. To keep this paper as self-contained as possible, we recall a sketch of this algorithm, proposed in~\cite{ChLaRa19}. The classical BFS uses a queue to store the visited vertices of the graph. Here, we initialize this queue by putting into it all the vertices of $H$ (instead of a single starting vertex). Each vertex $v \in V$ is labeled by three variables: $d(v)$ (distance to the gate), $f(v)$ (parent of $v$ in the BFS tree) and $\texttt{fib}(v)$ ($H$-gate of $v$). All vertices $x \in H$ are initialized with $d(x) = 0$, $f(x) = \texttt{NULL}$, and $\texttt{fib}(x) = x$. 
Then, we proceed the search as follows. Once a vertex $v$ is at the head of the queue, all not yet discovered neighbors $w$ of $v$ are inserted into the queue, and we fix $d(w) = d(v) + 1$, $f(w) = v$, $\texttt{fib}(w) = \texttt{fib}(v)$. At the end of the execution, each vertex $v$ is labeled by not only its gate $\texttt{fib}(v) = g_H(v)$ in $H$ but also the (unweighted) distance towards its gate $d(v)$. For the correctness of this BFS traversal, see Lemma 16 and Corollary 6 from~\cite{ChLaRa18}, the open access version of~\cite{ChLaRa19}.

\begin{lemma}[\cite{ChLaRa18,ChLaRa19}]
For any gated set $H$ of a graph $G$, one can compute, in linear time $O(m)=O(n\log n)$, all fibers $F_H[x]$, with $x\in H$, but also all distances from each $v \in V(G)$ to its $H$-gate.
\label{le:compute_gates}
\end{lemma}

Let us denote by   $(G,\omega)=((V,E),\omega)$ the \textit{weighted graph} $G$ when it is equipped with an  integer  non-negative weight function on its vertices   $\omega : V \rightarrow \mathbb{N}$. In brief, in this paper, expression \textit{weighted graph} refers to vertex-weighted graphs. Observe that the definitions of the structural notions we introduced above are independent from any weight consideration. Notions of interval, gate, convex and gated sets, and fiber only depend on unweighted distances.

Now, considering a weighted graph $(G,\omega)$, the distance is  generalized to a "\textit{weighted distance}" which is not symmetrical anymore: $d_{\omega}(u,v) = d(u,v) + \omega(v)$. In particular, $d_{\omega}(u,u) = \omega(u)$. When there is some ambiguity on the graph considered, we may write the weighted distance as $d_{(G,\omega)}(u,v)$. 
The \textit{weighted eccentricity} of a vertex $u \in V$ is the maximum weighted distance from this vertex, {\em i.e.} $\wecc{u}{(G,\omega)} = \max_{v \in V} d_{\omega}(u,v)$. In future sections, this notion of weighted distances/eccentricities will allow us to describe recursive algorithms on median graphs based on gated sets.

\subsection{Median graphs and $\Theta$-classes}

From now on, we focus on the family of graphs which is studied in this article: median graphs. 

\begin{definition}[Median graph]
A graph is \textit{median} if, for any triplet $(x,y,z)$ of distinct vertices, the set $I(x,y) \cap I(y,z) \cap I(z,x)$ contains exactly one vertex $m(x,y,z)$ called the \emph{median} of $(x,y,z)$.
\label{def:median}
\end{definition}

Trees, hypercubes, grids and squaregraphs~\cite{BaChEp10} are median graphs.
Median graphs are bipartite and do not contain any induced $K_{2,3}$~\cite{BaCh08,HaImKl11,Mu78}. The Cartesian product of two median graphs is also median~\cite{BrKlSk07}.

\begin{figure}[h]
\centering
\scalebox{0.7}{\begin{tikzpicture}


\node[draw, circle, minimum height=0.2cm, minimum width=0.2cm, fill=black] (P11) at (1,1) {};
\node[draw, circle, minimum height=0.2cm, minimum width=0.2cm, fill=black] (P12) at (1,2.5) {};

\node[draw, circle, minimum height=0.2cm, minimum width=0.2cm, fill=black] (P21) at (3,1) {};
\node[draw, circle, minimum height=0.2cm, minimum width=0.2cm, fill=black] (P22) at (3,2.5) {};
\node[draw, circle, minimum height=0.2cm, minimum width=0.2cm, fill=black] (P23) at (3,4) {};

\node[draw, circle, minimum height=0.2cm, minimum width=0.2cm, fill=black] (P31) at (5,1) {};
\node[draw, circle, minimum height=0.2cm, minimum width=0.2cm, fill=black] (P32) at (5,2.5) {};
\node[draw, circle, minimum height=0.2cm, minimum width=0.2cm, fill=black] (P33) at (5,4) {};

\node[draw, circle, minimum height=0.2cm, minimum width=0.2cm, fill=black] (P41) at (1.3,3.8) {};
\node[draw, circle, minimum height=0.2cm, minimum width=0.2cm, fill=black] (P42) at (1.3,5.3) {};
\node[draw, circle, minimum height=0.2cm, minimum width=0.2cm, fill=black] (P43) at (-0.7,3.8) {};

\node[draw, circle, minimum height=0.2cm, minimum width=0.2cm, fill=black] (P44) at (-0.7,2.0) {};

\node[draw, circle, minimum height=0.2cm, minimum width=0.2cm, fill=black] (P51) at (4.0,5.3) {};
\node[draw, circle, minimum height=0.2cm, minimum width=0.2cm, fill=black] (P52) at (6.0,5.3) {};

\node[draw, circle, minimum height=0.2cm, minimum width=0.2cm, fill=black] (P61) at (7,1) {};
\node[draw, circle, minimum height=0.2cm, minimum width=0.2cm, fill=black] (P62) at (7,2.5) {};

\node[draw, circle, minimum height=0.2cm, minimum width=0.2cm, fill=black] (P45) at (-0.7,5.3) {};

\node[draw, circle, minimum height=0.2cm, minimum width=0.2cm, fill=black] (P71) at (6.0,1.9) {};
\node[draw, circle, minimum height=0.2cm, minimum width=0.2cm, fill=black] (P72) at (6.0,3.4) {};
\node[draw, circle, minimum height=0.2cm, minimum width=0.2cm, fill=black] (P73) at (8.0,1.9) {};
\node[draw, circle, minimum height=0.2cm, minimum width=0.2cm, fill=black] (P74) at (8.0,3.4) {};


\draw[line width = 1.8pt, color = red, dotted] (P11) -- (P12);
\draw[line width = 1.4pt] (P11) -- (P21);
\draw[line width = 1.4pt] (P12) -- (P22);
\draw[line width = 1.8pt, color = red, dotted] (P21) -- (P22);

\draw[line width = 1.4pt] (P21) -- (P31);
\draw[line width = 1.4pt] (P22) -- (P32);
\draw[line width = 1.8pt, color = red, dotted] (P31) -- (P32);

\draw[line width = 1.4pt] (P22) -- (P23);
\draw[line width = 1.4pt] (P23) -- (P33);
\draw[line width = 1.4pt] (P32) -- (P33);

\draw[line width = 1.8pt, color = blue, loosely dashed] (P22) -- (P41);
\draw[line width = 1.8pt, color = blue, loosely dashed] (P12) -- (P43);
\draw[line width = 1.8pt, color = blue, loosely dashed] (P23) -- (P42);
\draw[line width = 1.4pt] (P41) -- (P43);
\draw[line width = 1.4pt] (P41) -- (P42);
\draw[line width = 1.4pt] (P12) -- (P44);

\draw[line width = 1.4pt] (P23) -- (P51);
\draw[line width = 1.4pt] (P33) -- (P52);
\draw[line width = 1.4pt] (P51) -- (P52);

\draw[line width = 1.4pt] (P31) -- (P61);
\draw[line width = 1.4pt] (P32) -- (P62);
\draw[line width = 1.8pt, color = red, dotted] (P61) -- (P62);
\draw[line width = 1.4pt] (P43) -- (P45);

\draw[line width = 1.8pt, color = red, dotted] (P71) -- (P72);
\draw[line width = 1.4pt] (P71) -- (P73);
\draw[line width = 1.4pt] (P72) -- (P74);
\draw[line width = 1.8pt, color = red, dotted] (P73) -- (P74);

\draw[line width = 1.4pt] (P31) -- (P71);
\draw[line width = 1.4pt] (P32) -- (P72);
\draw[line width = 1.4pt] (P61) -- (P73);
\draw[line width = 1.4pt] (P62) -- (P74);

\end{tikzpicture}}
\caption{Example of median graph $G$ with $d=3$. Two $\Theta$-classes are highlighted.}
\label{fig:median_example}
\end{figure}
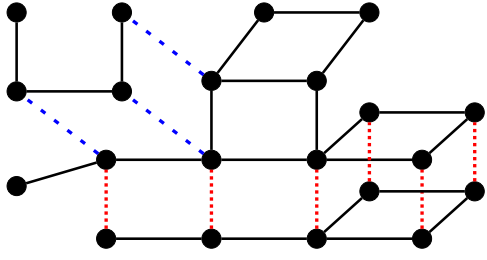

Given an integer $k \ge 1$, we denote by $Q_k$ the hypercube of dimension $k$. Inductively, graph $Q_1$ is the single-edge graph, and $Q_k$ is the Cartesian product of $Q_{k-1}$ and $Q_1$ for $k\ge 2$. In other words, $Q_k$ is obtained by taking graph $Q_{k-1}$ with a copy of it, and connecting each vertex with its own copy. For example, $Q_2$ is a square and $Q_3$ is the well-known 8-vertex cube of dimension $3$.

\begin{definition}[Dimension $d$]
The dimension $d = \mbox{dim}(G)$ of a median graph $G$ is the dimension of the largest hypercube contained in $G$ as an induced subgraph.
\label{def:dimension}
\end{definition}

In other words, $d = \mbox{dim}(G)$ means that $G$ admits $Q_d$ as an induced subgraph, but not $Q_{d+1}$. Median graphs with $d=1$ are exactly the trees. Median graphs with $d\le 2$ are called \textit{cube-free} median graphs. Figure~\ref{fig:median_example} provides us with an example of median graph with dimension $d = 3$. In general, since $Q_d$ contains $2^d$ vertices:
$$d = \dim(G) \le \log_2 n.$$



As stated previously, in general graphs, all gated subgraphs are convex. In median graphs, these two characterizations are equivalent.
\begin{lemma}[\cite{BaCh08,BeChChVa20}]
A subgraph of a median graph is gated if and only if  it is convex.
\end{lemma}

In fact, fibers are also gated.

\begin{lemma}[\cite{Ch01,ChLaRa19}]
For any gated set $H$ of a median graph $G$ and any vertex $x\in H$, the set $F_H[x]$ is gated.
\label{le:fibers_gated}
\end{lemma}

We recall the notion of $\Theta$-class~\cite{Dj73} and its implications on median graphs. We say that the edges $uv$ and $xy$ are in relation $\Theta_0$ if there is a square $uvyx$, where $uv$ and $xy$ are opposite. Notation $\Theta$ refers to the reflexive and transitive closure of relation $\Theta_0$. The classes of the equivalence relation $\Theta$ are denoted by $E_1,\ldots,E_q$ and the set of $\Theta$-classes is $\mathcal{E}(G)$.

\begin{definition}[$\Theta$-classes]
Two edges $uv$ and $u'v'$ belong to the same $\Theta$-class if there is a sequence $uv, u_1v_1, \ldots, u_rv_r, u'v'$ such that $u_iv_i$ and $u_{i+1}v_{i+1}$ are opposite edges of a square.
\label{def:classes}
\end{definition}
The size $q$ of $\mathcal{E}(G)$ can fluctuate: hypercube $Q_k$ admits $q=k$ $\Theta$-classes but $2^k$ vertices, while a tree has $q=n-1$ $\Theta$-classes (as many as edges). In any case, the $\Theta$-classes can be assigned to edges of a median graph in linear time.

\begin{lemma}[$\Theta$-classes in linear time~\cite{BeChChVa20}]
There exists an algorithm which computes the $\Theta$-classes $E_1,\ldots,E_q$ of a median graph in linear time $O(m) = O(n\log n)$.
\label{le:linear_classes}
\end{lemma}

In median graphs, each class $E_i$, $1\le i\le q$, is a perfect matching cutset and its two sides $H_i'$ and $H_i''$ admit convex characterizations. 

\begin{lemma}[Halfspaces of $E_i$~\cite{BeChChVa20,HaImKl99,Mu80}]
Let $G$ be any median graph and $E_i \in \mathcal{E}(G)$. For any $1\le i\le q$, the graph $G$ deprived of edges of $E_i$, {\em i.e.} $G\backslash E_i = (V,E\backslash E_i)$, has two connected components with respective vertex sets $H_i'$ and $H_i''$, called \emph{halfspaces}. Edges of $E_i$ form a matching. Halfspaces satisfy the following properties:
\begin{itemize}
\item Both $H_i'$ and $H_i''$ are convex/gated.
\item If $uv$ is an edge of $E_i$ with $u \in H_i'$ and $v \in H_i''$, then:\\ $H_i' = \set{x \in V: d(x,u) < d(x,v)}$\\$H_i'' = \set{x \in V: d(x,v) < d(x,u)}$.
\end{itemize}
\label{le:halfspaces}
\end{lemma}

We say a \textit{minority halfspace} is a halfspace, say $H_i'$ w.l.o.g, such that $\vert H_i' \vert < \vert H_i'' \vert$. The opposite halfspace $H_i''$ is thus naturally called a \textit{majority halfspace}. Moreover, we say that $H_i'$ and $H_i''$ are \textit{egalitarian halfspaces} if $\vert H_i' \vert = \vert H_i'' \vert$. In summary, a pair $(H_i',H_i'')$ is made up of either a minority-majority configuration or two egalitarian halfspaces.

A very nice and powerful characterization of median graphs is the construction by convex expansions proposed by Mulder~\cite{Mu78,Mu80}. Indeed, given a graph $G$ with two convex subgraphs $G_1$ and $G_2$ covering $G$, the \textit{convex expansion} $G'$ of $G$ is the graph obtained by adding both an isomorphic copy of $G_1 \cap G_2$ and a matching connecting the corresponding vertices between $G_1 \cap G_2$ and its copy. While $G_1 \cap G_2$ stays attached to $G_1$, its copy replaces it in $G_2$, and the matching connects both versions of $G_1 \cap G_2$. The apparition of this matching in the new graph $G'$ consists in the creation of a new $\Theta$-class. In summary, any median graph can be obtained by successive convex expansions starting from the single-vertex graph~\cite{Mu78}.

We denote by $\partial H_i'$ the subset of $H_i'$ containing the vertices which are adjacent to a vertex in $H_i''$: $\partial H_i' = N(H_i'')$. Put differently, the set $\partial H_i'$ is made up of vertices of $H_i'$ which are endpoints of edges in $E_i$. Symmetrically, set $\partial H_i''$ contains the vertices of $H_i''$ which are adjacent to $H_i'$. We say these sets are the \textit{boundaries} of halfspaces $H_i'$ and $H_i''$ respectively. A halfspace satisfying $H_i' = \partial H_i'$ is called a \textit{peripheral halfspace} and its associated $\Theta$-class is also called a \textit{peripheral $\Theta$-class}. As observed in~\cite{BeChChVa20} any median graph necessarily admits at least one peripheral $\Theta$-class. 

\begin{lemma}[Boundaries~\cite{BeChChVa20,HaImKl99,Mu80}]
Let $G$ be a median graph and $E_i \in \mathcal{E}(G)$. Both $\partial H_i'$ and $\partial H_i''$ are convex/gated. Also, the edges of $E_i$ define an isomorphism between $\partial H_i'$ and $\partial H_i''$.
\label{le:boundaries}
\end{lemma}

As a consequence, suppose $uv$ and $u'v'$ belong to $E_i$: if $uu'$ is an edge and belongs to class $E_j$, then $vv'$ is an edge too and it belongs to $E_j$. We pursue with another property related to the fact that median graphs $G$ are bipartite. The \textit{$v_0$-orientation} of the edges of $G$ according to some vertex $v_0 \in V(G)$ is such that, for any edge $uv$, the orientation is $\vv{uv}$ if $d(v_0,u) < d(v_0,v)$. Indeed, we cannot have $d(v_0,u) = d(v_0,v)$ as $G$ is bipartite. 

\begin{lemma}[Orientation~\cite{BeChChVa20}]
All edges of a median graph $G$ can be oriented according to any vertex $v_0 \in V(G)$.
\end{lemma}

Considering such orientation fixed, we might refer to the vertex $v_0$ as the \emph{canonical basepoint}.
Given two vertices $u,v \in V$, the set which contains the $\Theta$-classes separating $u$ from $v$ is called the \textit{signature} $\sigma_{u,v}$. For example, if $u \in H_i'$ and $v \in H_i''$, then $E_i \in \sigma_{u,v}$.

\begin{definition}[Signature $\sigma_{u,v}$~\cite{BeDuHa22}]
We say that the {\em signature} of the pair of vertices $u,v$, denoted by $\sigma_{u,v}$, is the set of classes $E_i$ such that $u$ and $v$ are separated in $G\backslash E_i$. In other words, $u$ and $v$ are in different halfspaces of $E_i$.
\label{def:signature}
\end{definition}

As stated in~\cite{BeHa21}, the signature of two vertices provides us with the composition, in terms of $\Theta$-classes, of any shortest $(u,v)$-path. 

\begin{lemma}[\cite{BeHa21}]
For any shortest $(u,v)$-path $P$, the edges in $P$ belong to classes in $\sigma_{u,v}$ and, for any $E_i \in \sigma_{u,v}$, there is exactly one edge of $E_i$ in path $P$. Conversely, a path containing at most one edge of each $\Theta$-class is a shortest path between its departure and its arrival.
\label{le:signature}
\end{lemma}

This lemma is a consequence of the convexity of halfspaces. Indeed, a shortest path that would pass through two edges of some $\Theta$-class $E_i$ would escape temporarily from an halfspace. This is not possible since any halfspace is convex (Lemma~\ref{le:halfspaces}).


\subsection{Orthogonal $\Theta$-classes, POFs and ladder sets}

We now present other notions on median graphs related to the \textit{orthogonality} of $\Theta$-classes~\cite{Ko09}. 

\begin{definition}[Orthogonal $\Theta$-classes~\cite{Ko09}]
We say that classes $E_i$ and $E_j$ are {\em orthogonal} (denoted by $E_i \perp E_j$) if there is a square $uvyx$ in $G$, where $uv,xy \in E_i$ and $ux,vy \in E_j$.
\end{definition}



We focus on the set of $\Theta$-classes which are pairwise orthogonal.

\begin{definition}[Pairwise Orthogonal Family (POF)~\cite{BeDuHa22}]
We say that a set of classes $X \subseteq \mathcal{E}(G)$ is a {\em POF} if for any pair $E_j,E_h \in X$, we have $E_j \perp E_h$.
\end{definition}

The empty set is considered as a POF, such as the singletons of elements of $\mathcal{E}(G)$. 
The notion of POF has natural connections with induced hypercubes in median graphs. As an extreme case, the whole set $\mathcal{E}(G)$ is a POF if and only if graph $G$ is a hypercube of dimension $\log_2 n$~\cite{Ko09,MoMuRo98}. Conversely, if $G$ does not admit any POF of size at least $2$, then it is a tree, as it means that there is no square in $G$. The following lemma states an important observation linking POFs and hypercubes.

\begin{lemma}[POFs and hypercubes~\cite{BeHa21}]
Let $X$ be a POF, $v \in V(G)$, and assume that for each $E_i \in X$, there is an edge of $E_i$ adjacent to $v$. There exists a hypercube $Q$ containing vertex $v$ and all edges of $X$ adjacent to $v$. The $\Theta$-classes of the edges of $Q$ are the $\Theta$-classes of $X$.
\label{le:pof_adjacent}
\end{lemma}

Furthermore, a natural bijection between the vertices of a median graph and its POFs was highlighted in the literature~\cite{BaChDrKo06,BaQuSaMa02}. As a consequence, the number of POFs is equal to $n = \vert V(G) \vert$.

\begin{lemma}[POF/vertex bijection~\cite{BaChDrKo06,BaQuSaMa02}]
Let $G$ be a median graph and $v_0 \in V(G)$ an arbitrary basepoint. We consider the $v_0$-orientation of $G$. Given a vertex $v \in V(G)$, let $N^-(v)$ be the set of edges going into $v$ and $\mathcal{E}^-(v) \subseteq \mathcal{E}(G)$ the $\Theta$-classes of the edges in $N^-(v)$. 
\begin{itemize}
\item For any vertex $v\in V(G)$, $\mathcal{E}^-(v)$ is a POF. Moreover, both $v$ and the edges of $N^-(v)$ belong to an induced hypercube whose edges are in the $\Theta$-classes of $\mathcal{E}^-(v)$. 
\item For any POF $X$, there is an unique vertex $v_X$ such that $\mathcal{E}^-(v_X) = X$. Vertex $v_X$ is the closest-to-$v_0$ vertex $v$ such that $X \subseteq \mathcal{E}^-(v)$. As $v_X$ and $N^-(v_X)$ belong to a common induced hypercube, any POF $X$ verifies $\vert X \vert \le d$.
\end{itemize}
\label{le:pof_hypercube}
\end{lemma}

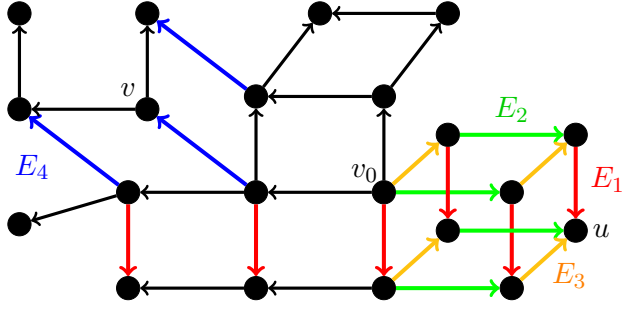
\begin{figure}[h]
\centering
\scalebox{0.85}{\begin{tikzpicture}


\node[draw, circle, minimum height=0.2cm, minimum width=0.2cm, fill=black] (P11) at (1,1) {};
\node[draw, circle, minimum height=0.2cm, minimum width=0.2cm, fill=black] (P12) at (1,2.5) {};

\node[draw, circle, minimum height=0.2cm, minimum width=0.2cm, fill=black] (P21) at (3,1) {};
\node[draw, circle, minimum height=0.2cm, minimum width=0.2cm, fill=black] (P22) at (3,2.5) {};
\node[draw, circle, minimum height=0.2cm, minimum width=0.2cm, fill=black] (P23) at (3,4) {};

\node[draw, circle, minimum height=0.2cm, minimum width=0.2cm, fill=black] (P31) at (5,1) {};
\node[draw, circle, minimum height=0.2cm, minimum width=0.2cm, fill=black] (P32) at (5,2.5) {};
\node[draw, circle, minimum height=0.2cm, minimum width=0.2cm, fill=black] (P33) at (5,4) {};

\node[draw, circle, minimum height=0.2cm, minimum width=0.2cm, fill=black] (P41) at (1.3,3.8) {};
\node[draw, circle, minimum height=0.2cm, minimum width=0.2cm, fill=black] (P42) at (1.3,5.3) {};
\node[draw, circle, minimum height=0.2cm, minimum width=0.2cm, fill=black] (P43) at (-0.7,3.8) {};

\node[draw, circle, minimum height=0.2cm, minimum width=0.2cm, fill=black] (P44) at (-0.7,2.0) {};

\node[draw, circle, minimum height=0.2cm, minimum width=0.2cm, fill=black] (P51) at (4.0,5.3) {};
\node[draw, circle, minimum height=0.2cm, minimum width=0.2cm, fill=black] (P52) at (6.0,5.3) {};

\node[draw, circle, minimum height=0.2cm, minimum width=0.2cm, fill=black] (P61) at (7,1) {};
\node[draw, circle, minimum height=0.2cm, minimum width=0.2cm, fill=black] (P62) at (7,2.5) {};

\node[draw, circle, minimum height=0.2cm, minimum width=0.2cm, fill=black] (P45) at (-0.7,5.3) {};

\node[draw, circle, minimum height=0.2cm, minimum width=0.2cm, fill=black] (P71) at (6.0,1.9) {};
\node[draw, circle, minimum height=0.2cm, minimum width=0.2cm, fill=black] (P72) at (6.0,3.4) {};
\node[draw, circle, minimum height=0.2cm, minimum width=0.2cm, fill=black] (P73) at (8.0,1.9) {};
\node[draw, circle, minimum height=0.2cm, minimum width=0.2cm, fill=black] (P74) at (8.0,3.4) {};

\node[scale=1.2,color = blue] at (-0.5,2.9) {$E_4$};
\node[scale=1.2,color = red] at (8.5,2.7) {$E_1$};
\node[scale=1.2,color = black!20!green] at (7.0,3.8) {$E_2$};
\node[scale=1.2,color = orange] at (7.9,1.2) {$E_3$};

\node[scale=1.2] at (4.7,2.8) {$v_0$};
\node[scale=1.2] at (1.0,4.1) {$v$};
\node[scale=1.2] at (8.4,1.9) {$u$};


\draw[line width = 1.8pt, color = red, <-] (P11) -- (P12);
\draw[line width = 1.4pt, <-] (P11) -- (P21);
\draw[line width = 1.4pt, <-] (P12) -- (P22);
\draw[line width = 1.8pt, color = red, <-] (P21) -- (P22);

\draw[line width = 1.4pt, <-] (P21) -- (P31);
\draw[line width = 1.4pt, <-] (P22) -- (P32);
\draw[line width = 1.8pt, color = red, <-] (P31) -- (P32);

\draw[line width = 1.4pt, ->] (P22) -- (P23);
\draw[line width = 1.4pt, <-] (P23) -- (P33);
\draw[line width = 1.4pt, ->] (P32) -- (P33);

\draw[line width = 1.8pt, color = blue, ->] (P22) -- (P41);
\draw[line width = 1.8pt, color = blue, ->] (P12) -- (P43);
\draw[line width = 1.8pt, color = blue, ->] (P23) -- (P42);
\draw[line width = 1.4pt, ->] (P41) -- (P43);
\draw[line width = 1.4pt, ->] (P41) -- (P42);
\draw[line width = 1.4pt, ->] (P12) -- (P44);

\draw[line width = 1.4pt, ->] (P23) -- (P51);
\draw[line width = 1.4pt, ->] (P33) -- (P52);
\draw[line width = 1.4pt, <-] (P51) -- (P52);

\draw[line width = 1.8pt, color = green, ->] (P31) -- (P61);
\draw[line width = 1.8pt, color = green, ->] (P32) -- (P62);
\draw[line width = 1.8pt, color = red, <-] (P61) -- (P62);
\draw[line width = 1.4pt, ->] (P43) -- (P45);

\draw[line width = 1.8pt, color = red, <-] (P71) -- (P72);
\draw[line width = 1.8pt, color = green, ->] (P71) -- (P73);
\draw[line width = 1.8pt, color = green, ->] (P72) -- (P74);
\draw[line width = 1.8pt, color = red, <-] (P73) -- (P74);

\draw[line width = 1.8pt, color = yellow!50!orange, ->] (P31) -- (P71);
\draw[line width = 1.8pt, color = yellow!50!orange, ->] (P32) -- (P72);
\draw[line width = 1.8pt, color = yellow!50!orange, ->] (P61) -- (P73);
\draw[line width = 1.8pt, color = yellow!50!orange, ->] (P62) -- (P74);

\end{tikzpicture}}
\caption{The $v_0$-orientation of some median graph $G$ and some of its $\Theta$-classes. For example, $\mathcal{E}^-(u) = \{ E_1,E_2,E_3 \}$.}
\label{fig:vertices_pof}
\end{figure}

An example is given in Figure~\ref{fig:vertices_pof} with the same median graph than in Figure~\ref{fig:median_example}. The $v_0$-orientation of this graph is represented, with a sample of four $\Theta$-classes. With the notation used in the previous lemma, we have: $\mathcal{E}^-(v_0) = \emptyset$,  $\mathcal{E}^-(v) = \{ E_4 \}$, and  $\mathcal{E}^-(u) = \{ E_1, E_2, E_3 \}$. One can check that each vertex admits its own incoming set of $\Theta$-classes. 
Due to this POF/vertex bijection, POFs of a median graph can be enumerated in linear time~\cite{BaQuSaMa02,Ko09}. Furthermore, as with the $v_0$-orientation, at most $d$ arcs enter in each vertex ($\mathcal{E}^-(u)$ is a POF), median graphs are relatively sparse: $m \le dn \le n\log_2 n$.


We focus on another notion strongly related to POFs, defined in~\cite{BeHa21}, called \textit{ladder set}.

\begin{definition}[Ladder set~\cite{BeHa21}]
Given two vertices $u\neq v$ of a median graph $G$, the ladder set $L_{u,v}$ is the set of $\Theta$-classes which are both adjacent to $u$ and belong to $\sigma_{u,v}$:
$$L_{u,v} = \left\{E_i \in \sigma_{u,v} : u \in \partial H_i' \cup \partial H_i''\right\}.$$
\label{def:ladder}
\end{definition}
For example, in Figure~\ref{fig:vertices_pof}, we have $L_{v,u} = L_{v,v_0} = \{E_4\}$ and $L_{u,v} = L_{v_0,u} = \{E_1,E_2,E_3\}$. In fact, the $\Theta$-classes of a given ladder set are pairwise orthogonal.

\begin{lemma}[\cite{BeHa21}]
Any ladder set $L_{u,v}$ is a POF.
\label{le:ladder_POF}
\end{lemma}

Less formally, the ladder set $L_{u,v}$ provides us with the $\Theta$-classes of the induced hypercube containing $u$ covered by the set of all shortest $(u,v)$-paths.

Given a basepoint $v_0$, all ladder sets $L_{v_0,v}$, with $v \in V(G) \setminus \{v_0\}$ can be enumerated in quasilinear time thanks to a BFS starting at $v_0$. Such an algorithm is evoked in~\cite{BeDuHa22} but is not clearly stated, hence we do so.

\begin{lemma}[\cite{BeDuHa22}]
Given a basepoint $v_0$ of some median graph $G$, the list of all ladder sets $L_{v_0,v}$ with $v\neq v_0$ can be enumerated in quasilinear time $O(n\log^2 n)$.
\label{le:ladder_linear}
\end{lemma}
\begin{proof}
Initialize the queue with the starting vertex $v_0$. Label all vertices with a set $\lad(x) = \emptyset$. Once a vertex $v$ is at the head of the queue, all its not yet discovered neighbors $w$ of $v$ are inserted into the queue, and we fix $\lad(w) = \lad(v) \cup E_i$ if $vw \in E_i$ and $v_0$ is adjacent to $E_i$, otherwise $\lad(w) = \lad(v)$. Labels $\lad$ exactly compute ladder sets from Definition~\ref{def:ladder}: $\lad(v)$ contains the $\Theta$-classes adjacent to $v_0$ which separates $v$ from $v_0$. As ladder sets are POFs which contain the identity of at most $\log_2 n$ $\Theta$-classes, the size needed to store each label is at most $(\log_2 n)^2$. BFS execution together with the writing of labels gives a total running time $O(n\log^2 n)$ since $m = O(n\log n)$.
\end{proof}

\subsection{Median set and majority rule}

It was  recently proposed in ~\cite{BeChChVa20} a linear-time algorithm which computes the median set of median graphs. We recall here some key observations of this article but also previous works that will be useful for us. We begin with the definition of a median set.

\begin{definition}[Median vertex and set]
Given a graph $G$, a \emph{median vertex} of $G$ is a vertex $u$ which minimizes $\Gamma(u) = \sum_{v\in V(G)} d(u,v)$. The median set $\med(G)$ is the set containing all median vertices.
\end{definition}

In median graphs, the median set admits an interesting characterization related to the $\Theta$-classes. Indeed, any vertex which belongs to at least one minority halfspace is not median.

\begin{lemma}[Majority rule~\cite{BaBa84}]
$\med(G)$ is the intersection of all majority halfspaces. It coincides with the interval of a diametral pair of its vertices.
\label{le:majority}
\end{lemma}

The majority rule is a key tool for the algorithm presented in~\cite{BeChChVa20}. The first part consists in computing the cardinality of each halfspace of the input median graph $G$. This is based on a \textit{peripheral peeling} that we describe briefly. Consider $G$ where all vertex weights are fixed to $1$. Pick up some peripheral $\Theta$-class $E_i$ : first retrieve the cardinality of its peripheral halfspace (say $H_i' = \partial H_i'$) by summing up all weights of $H_i'$, second transfer the weights of vertices in $H_i'$ to their $E_i$-neighbors. Remove $H_i'$ from the current graph and recurse the process on another peripheral $\Theta$-class.

\begin{lemma}[Halfspaces sizes~\cite{BeChChVa20}]
    Given a median graph $G$, there is a combinatorial algorithm computing in linear time $O(m) = O(n\log n)$ all triplets $(E_i,\vert H_i'\vert, \vert H_i'' \vert)$ where $E_i \in \mathcal{E}(G)$.
    \label{le:enum_halfspaces}
\end{lemma} 

Once the cardinality of each halfspace is known, the majority rule allows them to retrieve the median set $\med(G)$. The second part consists simply in orienting edges $uv$ (say $uv \in E_j$ and $u \in H_j'$ w.l.o.g)  such that $v$ is the head of the arc iff $\vert H_j'' \vert > \vert H_j' \vert$. From Lemma~\ref{le:majority}, the median set coincides with the sinks of this partially directed graph.

\begin{corollary}[\cite{BeChChVa20}]
    $\med(G)$ can be computed in linear time $O(m) = O(n\log n)$.
    \label{co:median_linear}
\end{corollary}

\section{Exploiting halfspaces for the computation of eccentricities} \label{sec:first_step}

Our objective is to  determine the weighted eccentricities of a weighted median graph $(G,\omega)$. Please note that notation $G$ naturally refers to the same graph without any weight consideration. To achieve our goal, we introduce the notion of balanced $\Theta$-classes. Our algorithm will distinguish two cases: either $G$ contains a balanced $\Theta$-class or not. If it is the case, then we will pick up such a $\Theta$-class and retrieve the eccentricities of $(G,\omega)$ by computing recursively the eccentricities of the two halfspaces $(G[H_i'],\omega)$ and $(G[H_i''],\omega)$. 

\subsection{Balanced $\Theta$-classes} \label{subsec:balanced}

As stated in Section~\ref{sec:notation}, the $\Theta$-classes are natural separators for a median graph. Consequently, the ratio between the size of the two halfspaces is a potential tool to design divide-and-conquer procedures on this family of graphs. We begin with the definition of $f$-balanced $\Theta$-classes.

\begin{definition}[Balanced and unbalanced $\Theta$-classes]
Let $\mathbb{N}^* = \mathbb{N}\setminus \{0\}$ be the set of positive integers. Given some median graph $G$ and a function $f : \mathbb{N}^* \rightarrow \mathbb{N}^*$, a $\Theta$-class $E_i$ of $G$ is $f$-balanced if:
\begin{equation}\min\{\vert H_i' \vert, \vert H_i'' \vert\} \ge \frac{\vert V(G) \vert}{f(\vert V(G) \vert)}.
\label{eq:balanced}
\end{equation}
Conversely, a $\Theta$-class is $f$-unbalanced if it is not $f$-balanced, so formally either $\vert H_i' \vert < \frac{\vert V(G) \vert}{f(\vert V(G) \vert)}$ or $\vert H_i'' \vert < \frac{\vert V(G) \vert}{f(\vert V(G) \vert)}$.
\label{def:balanced}
\end{definition}

The halfspaces of a $\Theta$-class $E_i$ form a bipartition of $V(G)$. Hence, being $f$-unbalanced means that one halfspace has size less than $\frac{n}{f(n)}$, with $n = \vert V(G) \vert$, while the other halfspace is large with cardinality at least $n(1-\frac{1}{f(n)})$. Star graphs are the most natural examples of median graphs without balanced $\Theta$-classes: for each $\Theta$-class, its minority halfspace has size $1$. Observe that the case $f$ being a constant function (for example, $f(n) = 3$ for all integers $n$) is the classical sense given to a \textit{balanced separator} in graphs, as defined originally by~\cite{LiTa79}. In the remainder of this article, we will focus more particularly on $f$-balanceness, with $f$ being a logarithmic function. 

It should be noticed that certain median graphs do not admit any $f$-balanced $\Theta$-classes, for some given function $f$. 

\begin{definition}[Unbalanced median graphs]
Let $\mathcal{U}_f$ be the family of all median graphs $G$ satisfying the following property: all $\Theta$-classes of $G$ are $f$-unbalanced. 
\label{def:unbalanced_median}
\end{definition}

For Sections~\ref{sec:first_step} and~\ref{sec:unbalanced}, we fix a specific function $f$ and the notion of \textit{balanced}/\textit{unbalanced} $\Theta$-class will be naturally associated with this function : let $f = 2\log$, {\em i.e.} the function $f : n \rightarrow 2\log(n)$. We denote by $\ulog$ the family of median graphs which do not admit any $(2\log)$-balanced $\Theta$-classes. Fixing $f = 2\log$, which is less restrictive than a constant function, is crucial to make our algorithm work: this choice will be justified in the proof of Theorem~\ref{th:main_ecc} (as presented in section \ref{sec:intro}).

Checking whether a median graph admits (or not) a balanced $\Theta$-class can be achieved in linear time $O(n\log n)$, according to Lemma~\ref{le:enum_halfspaces}. 
As a consequence, we are able in linear time either to pick up a balanced $\Theta$-class (for any balance criterion $f$) or answer that the input median graph contains only unbalanced ones. It consists simply in applying the algorithm of Lemma~\ref{le:enum_halfspaces} and returning a balanced $\Theta$-class when Equation~\eqref{eq:balanced} is satisfied.

The idea of our future recursive scheme is to identify whether a balanced $\Theta$-class of the input median graph $G$ exists. If it is the case, we compute recursively the weighted eccentricities of its halfspaces in order to retrieve the weighted eccentricities of the whole graph. Otherwise, we fall into the case $G \in \ulog$ which will be treated in Section~\ref{sec:unbalanced}.

We state an analytical result on some integer sequences which will be a key win-win argument for this global recursive process. The consequence of the following lemma is that the depth of the recursive tree, built upon balanced $\Theta$-classes separation, is at most poly-logarithmic.

\begin{lemma}
Let $(s_n)$ be a sequence of integers such that $s_0$ is a positive integer, and let $\lambda \ge 2$ be a positive real number such that: $$s_{n+1} = \left\{ \begin{array}{ll}
        \lfloor s_n\left(1-\frac{1}{\lambda\log(s_n)}\right) \rfloor & \mbox{ if } s_n > 2\\
        s_n & \mbox{ if } s_n \le 2
\end{array}
\right.$$
The total stopping time $\tau$ of sequence $s_n$, {\em i.e.} the minimum positive integer $\tau$ such that $s_{\tau} = s_{\tau+1}$, verifies $\tau \le \lambda(\log(s_0))^2$. Moreover, $s_{\tau}\in \{1,2\}$.
\label{le:log_decrease}
\end{lemma}
\begin{proof}
Consider a finite subsequence $s_0,s_1,\ldots,s_{n}$ and assume that $s_n > 2$.
By definition, it is monotonically decreasing since $\lambda \log(s_i) > 1$ for any $0\le i\le n$.

Then, $s_n \le s_0 \prod_{i=0}^{n-1}\left(1-\frac{1}{\lambda\log(s_i)}\right) \le s_0 \left(1-\frac{1}{\lambda\log(s_n)}\right)^n$.
Assume by contradiction that $n \ge \lambda(\log(s_0))^2$, then by exploiting the fact that $(1-\frac{1}{x})^x \le \frac{1}{e}$ for $x \ge 2$,
\begin{flalign*}
s_n \le s_0 \left(1-\frac{1}{\lambda\log(s_n)}\right)^{\lambda(\log(s_0))^2} 
\le s_0 \left(1-\frac{1}{\lambda\log(s_n)}\right)^{\lambda\log(s_n)\log(s_0)}
\le s_0 e^{-\log(s_0)} = 1.
\end{flalign*}
This yields a contradiction since we assumed $s_n > 2$. Hence, for $n \ge \lambda(\log(s_0))^2$, $s_n \le 2$, so the total stopping time verifies $\tau \le \lambda(\log(s_0))^2$. Now, we verify that $s_{\tau} > 0$. As $s_{\tau -1} \ge 3$, we have $s_{\tau} = \lfloor s_{\tau -1}\left(1-\frac{1}{\lambda\log(s_{\tau -1})}\right)\rfloor \ge \lfloor 3(1-\frac{1}{2\log(3))}\rfloor \ge 1$. Therefore, $s_{\tau} \in \{1,2\}$.
\end{proof}

\subsection{Retrieving all eccentricities thanks to $\Theta$-classes}

The following theorem is crucial for our algorithm: it states that, given a $\Theta$-class and the weighted eccentricity of each vertex inside its induced halfspace, we can re-assemble all weighted eccentricities of $(G,\omega)$ in linear time.

\begin{theorem}
Let $(G,\omega)$ be a weighted median graph and $E_i$ one of its $\Theta$-classes. Assume that:
\begin{itemize}
\item all weighted eccentricities of $(G[H_i'],\omega)$ are known,
\item all weighted eccentricities of $(G[H_i''],\omega)$ are known.
\end{itemize}
Then, one can compute all weighted eccentricities of $(G,\omega)$ in linear time $O(\vert E(G)\vert) = O(n\log n)$.
\label{th:balanced_recursion}
\end{theorem}
\begin{proof}
At the beginning of our computation, every vertex of $V(G)$ is labeled with a  weighted distance: the vertices $u' \in H_i'$ are labeled with their weighted eccentricity in $G[H_i']$, formally $\wecc{u'}{(H_i',\omega)}$. Similarly, each vertex $u'' \in H_i''$ is labeled with $\wecc{u''}{(H_i'',\omega)}$. 

For any vertex $u' \in H_i'$, we determine its $H_i''$-gate $g(u') \in H_i''$ and, conversely, for any vertex $u'' \in H_i''$, we determine its $H_i'$-gate $g(u'')$. This operation can be achieved in total $O(n\log n)$ time, as recalled in Lemma~\ref{le:compute_gates}, by launching two BFSs: one with a starting queue made up of $H_i'$, and one with a starting queue $H_i''$. We compute, for any $v' \in \partial H_i'$ (resp. $v'' \in \partial H_i''$) its open fiber $F_{H_i'}(v')$ in $H_i''$ (resp. $F_{H_i''}(v'')$ in $H_i'$). Through this BFS, we store the unweighted distance from any vertex to its gate: we obtain all values $d(u',g(u'))$ (resp. $d(u'',g(u''))$ also in linear time.

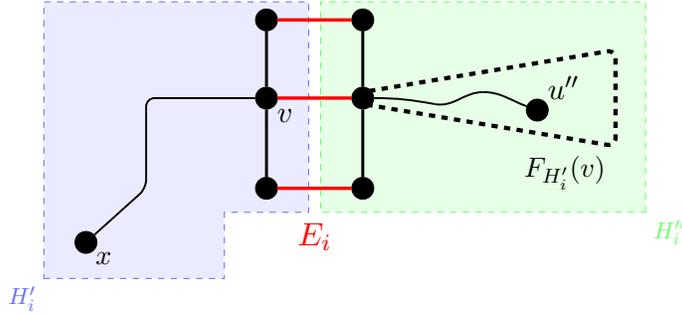
\begin{figure}[h]
\centering
\scalebox{0.8}{\begin{tikzpicture}


\draw [dashed, color = white!40!blue, fill = white!92!blue] (0,1.0) -- (3.0,1.0) -- (3.0,2.1) -- (4.4,2.1) -- (4.4,5.6) -- (0,5.6) -- (0,1.0) node[below left] {$H_i'$};
\draw [dashed, color = white!40!green, fill = white!90!green] (10.0,2.1) -- (10.0,5.6) -- (4.6,5.6) -- (4.6,2.1) -- (10.0,2.1) node[below right] {$H_i''$};


\node[draw, circle, minimum height=0.2cm, minimum width=0.2cm, fill=black] (P11) at (0.7,1.6) {};
\draw (P11) node[below right,scale =1.2] {$x$};

\node[draw, circle, minimum height=0.2cm, minimum width=0.2cm, fill=black] (P22) at (3.7,2.5) {};
\node[draw, circle, minimum height=0.2cm, minimum width=0.2cm, fill=black] (P23) at (3.7,4) {};
\draw (P23) node[below right,scale = 1.3] {$v$};
\node[draw, circle, minimum height=0.2cm, minimum width=0.2cm, fill=black] (P24) at (3.7,5.3) {};

\node[draw, circle, minimum height=0.2cm, minimum width=0.2cm, fill=black] (P32) at (5.3,2.5) {};
\node[draw, circle, minimum height=0.2cm, minimum width=0.2cm, fill=black] (P33) at (5.3,4) {};
\draw [rounded corners, line width = 2pt, dashed] (5.3,4.1) -- (5.3,3.9) -- (9.5,3.2) -- (9.5,4.8) -- (5.3,4.1);
\draw (9.5,3.2) node[below left,scale =1.2] {$F_{H_i'}(v)$};
\node[draw, circle, minimum height=0.2cm, minimum width=0.2cm, fill=black] (P34) at (5.3,5.3) {};

\node[draw, circle, minimum height=0.2cm, minimum width=0.2cm, fill=black] (P4) at (8.2,3.8) {};
\draw (P4) node[above right,scale = 1.3] {$u''$};
\draw[line width = 0.9pt,out=0,in=170] (P33)to(6.5,3.9);
\draw[line width = 0.9pt,out=-10,in=180] (6.5,3.9)to(7.3,4.1);
\draw[line width = 0.9pt,out=0,in=160] (7.3,4.1)to(P4);

\draw[rounded corners, line width = 0.9pt] (P23) -- (1.7,4) -- (1.7,2.5) -- (P11);


\draw[line width = 1.6pt, color = red] (P22) -- (P32);
\draw[line width = 1.6pt, color = red] (P23) -- (P33);
\draw[line width = 1.6pt, color=red] (P24) -- (P34);

\draw[line width = 1.4pt] (P22) -- (P23);
\draw[line width = 1.4pt] (P32) -- (P33);
\draw[line width = 1.4pt] (P23) -- (P24);
\draw[line width = 1.4pt] (P33) -- (P34);


\node[scale=1.4, color = red] at (4.5,1.7) {$E_i$};

\end{tikzpicture}}
\caption{The largest weighted distance from $u'' \in H_i''$ to some vertex $x \in H_i'$ is given by the (unweighted) distance from $u''$ to its gate $v = g_{H_i'}(u'')$ in addition with the weighted eccentricity of $v$: $\wecc{v}{(H_i',\omega)} = d(v,x)$.}
\label{fig:fibers}
\end{figure}

One can retrieve at this moment the weighted eccentricity of $u' \in H_i'$ (resp. $u'' \in H_i''$ with the same arguments). Indeed, the farthest vertex from $u'$ (in the weighted sense) is either in $H_i'$ or in $H_i''$. If it belongs to $H_i'$, as this halfspace is convex, the weighted eccentricity of $u'$ is its label $\wecc{u'}{(H_i',\omega)}$. Else, as $H_i''$ is gated, for any vertex $v'' \in H_i''$, there is a shortest path from $u'$ to $v''$ passing through the $H_i''$-gate $g(u')$. Conversely, any shortest path induced in $H_i''$ from $g(u')$ to some $v''$, concatenated with a shortest $(u',g(u'))$-path, produces a shortest $(u',v'')$-path, since $d(u',g(u')) + d(g(u'),v'') = d(u',v'')$
. Hence, the weighted eccentricity of $u'$ can be decomposed as the sum of $d(u',g(u'))$ (distance to the gate) with the weighted eccentricity of $g(u')$ in $H_i''$. 
Formally, $$\wecc{u'}{(G,\omega)} = \max \left\{\wecc{u'}{(H_i',\omega)}, ~d(u',g(u')) + \wecc{g(u')}{(H_i'',\omega)}\right\}.$$ 
Figure~\ref{fig:fibers} illustrates this formula with a vertex $u'' \in H_i''$ instead of $u' \in H_i'$.

The whole procedure, consisting first in computing the gates $g(u)$ and second in applying the latter formula, takes time $O(n\log n)$.
\end{proof}

Theorem \ref{th:balanced_recursion} immediately yields a recursive algorithm scheme based on $\Theta$-classes.
But its complexity is not necessarily subquadratic.

When our input graph $G$ has a balanced $\Theta$-class $E_i$, one can recursively compute the weighted eccentricities of the two induced halfspaces $G[H_i']$ and $G[H_i'']$ and then retrieve the weighted eccentricities of $G$ with an extra linear time. With some ``ideal'' instance, we could always find such a balanced $\Theta$-class and use this recursive scheme until falling into the trivial case of a single weighted vertex. Observe that, with such a utopian situation, the total running time would be quasilinear: the depth of the recursive tree is poly-logarithmic thanks to Lemma~\ref{le:log_decrease}. Indeed, fixing $s_0=n=\vert V(G) \vert$ and $\lambda  = 2$, value $s_n$ gives an upper bound of the number of vertices remaining in each branch of the recursive tree after $n$ calls. 

Nevertheless, we might find at some moment an induced subgraph of $G$ which has no balanced $\Theta$-class, in other words which belongs to $\ulog$. Obviously, even the input $G$ could belong to this family. Then the time complexity of the recursive scheme could be quadratic. 

As a conclusion of this section, we observe that, to pursue the description of our recursive algorithm, we need to focus on the case of weighted median graphs which admit only unbalanced $\Theta$-classes.

\section{Computation of eccentricities on median graphs without balanced $\Theta$-classes} \label{sec:unbalanced}

The study of median graphs with only unbalanced $\Theta$-classes is of interest since, as we prove it now, the existence of a quasilinear-time algorithm finding all weighted eccentricities for the subfamily $\ulog$ of median graphs will imply the same result for the whole family of median graphs. In Sections~\ref{subsec:slices} and~\ref{subsec:peeling}, we fix $(G,\omega) \in \ulog$.

\subsection{Slice decomposition} \label{subsec:slices}

We now pursue with new concepts about $\Theta$-classes. As all $\Theta$-classes $E_i$ of $G \in \ulog$ are unbalanced, all of them admit both a minority halfspace denoted from now on by $H_i'$ (the halfspace with the smaller size), and a majority halfspace denoted by $H_i''$. From Definition~\ref{def:balanced}, we have $\vert H_i' \vert \le \frac{n}{2\log(n)}$ with $n = \vert V(G) \vert$. Also, we assume that graph $G \in \ulog$ has at least three vertices: $n\ge 3$ and $\log n \ge 1$ (case $n\le 2$ will be treated trivially).

No graph of $\ulog$ contains an egalitarian halfspace since all its $\Theta$-classes are unbalanced. Consequently, these graphs admit a unique median vertex.

\begin{lemma}[Unique median vertex in $G\in \mathcal{U}_f$]
For any $n$-vertex median graph $G\in \mathcal{U}_f$ where $f(n) \ge 2$, there exists a unique vertex $v_0$ such that, for any $\Theta$-class $E_i$, this vertex $v_0$ belongs to the majority halfspace of $E_i$.
\label{le:unique_median}
\end{lemma}

\begin{proof}
From Lemma~\ref{le:majority}, the median set of a graph $G \in \mathcal{U}_f$ is the intersection of all majority halfspaces and there exists $u,v \in V(G)$ such that $\med(G) = I(u,v)$. Assume that $u\neq v$. As $\med(G)$ is connected, consider two adjacent vertices $x,y \in \med(G)$ and let $xy \in E_i$. The $\Theta$-class $E_i$ admits two egalitarian halfspaces since, otherwise, either $x$ or $y$ would not be a median vertex. We have a contradiction: $E_i$ is clearly balanced, but it should admit a halfspace of size strictly smaller than $\frac{n}{f(n)} \le \frac{n}{2}$.
\end{proof}

From now on, this unique median vertex $v_0$  will be taken as the basis of the orientation of $G$. We consider a second vertex $\umax$ (potentially equal to $v_0$) which is a farthest vertex from $v_0$.

\begin{definition}
Let $\umax$ be a vertex such that $d_{\omega}(v_0,\umax) = \wecc{v_0}{(G,\omega)}$. If several candidates for $\umax$ exist, then select one arbitrarily, except in one case: when $v_0$ itself is a candidate for $\umax$, then fix $\umax = v_0$.
\label{def:umax}
\end{definition}

We begin with a straightforward observation: if $\umax = v_0$, then $v_0$ is the farthest vertex from any $u \neq v_0$.

\begin{lemma}
If $v_0 = \umax$, then for any vertex $u \in V(G)\setminus \{v_0\}$, $\wecc{u}{(G,\omega)} = d_{\omega}(u,v_0)$.
\label{le:v0farthest}
\end{lemma}
\begin{proof}
By definition of $\umax$, we have $\omega(v_0) \ge d_{\omega}(v_0,u)$. Consider any vertex $v \neq v_0$. By triangular inequality, $d(u,v) \le d(u,v_0) + d(v_0,v)$, so $d_{\omega}(u,v) \le d(u,v_0) + d(v_0,v) + \omega(v) = d(u,v_0) + d_{\omega}(v_0,v)$. But $d_{\omega}(v_0,v) \le \omega(v_0)$ which implies that $d_{\omega}(u,v) \le d_{\omega}(u,v_0)$. As a consequence, $v_0$ is the farthest vertex from any $u \neq v_0$.
\end{proof}

A natural consequence of the previous lemma is that, if $\umax = v_0$, then the eccentricities of $(G,\omega) \in \ulog$ can be computed in linear time in a very simple way.

\begin{corollary}
    Let $(G,\omega) \in \ulog$ with $\umax = v_0$. All eccentricities of $(G,\omega)$ can be computed in linear time $O(m) = O(n\log n)$ thanks to a BFS starting at $v_0$.
    \label{co:casev0}
\end{corollary}
\begin{proof}
  Executing a BFS with departure vertex $v_0$ allows us to obtain all unweighted distances between $v_0$ and all other vertices. We know from Lemma~\ref{le:v0farthest} that, in our case, for any $u \neq v_0$, $\wecc{u}{(G,\omega)} = d_{\omega}(u,v_0) = d(u,v_0) + \omega(v_0)$. Moreover, $\wecc{v_0}{(G,\omega)} = \omega(v_0)$ since $\umax = v_0$. In summary, after computing all unweighted distances from $v_0$ in linear time, one can retrieve all eccentricities of $(G,\omega)$ also in linear time as a second step.
\end{proof}

From now on, we can assume that $\umax \neq v_0$ as otherwise the weighted eccentricities can be computed trivially. We define $L(G)$ as a ladder set $L_{v_0,\umax}$ between $v_0$ and one farthest-to-$v_0$ vertex $\umax$.

\begin{definition}[Wide ladder]
We denote by $L(G)$ the \emph{wide ladder} of $G$, which is the POF containing the $\Theta$-classes $E_i$ adjacent to $v_0$ such that $\umax$ belongs to their minority halfspace, {\em i.e.} $\umax \in H_i'$. Formally, $L(G) = L_{v_0,\umax}$ and $\ell = \vert L(G) \vert$ is its size.
\label{def:wide_ladder}
\end{definition}

For sake of simplicity, we modify the indices of the $\Theta$-classes such that: $$L(G) = \{E_1,E_2,\ldots,E_{\ell}\}.$$ Observe that this change has no impact on the previous results since the indices of $\Theta$-classes played no role yet.

Our technique to handle median graphs without balanced $\Theta$-classes relies on the following crucial observation. All vertices $v$ of $G$ such that $L_{v_0,v} \cap L(G) = \emptyset$ form, due to the unbalancedness of the $\Theta$-classes, a large set of vertices. Moreover, we can directly deduce their weighted eccentricity which is $d_{\omega}(v,\umax)$. As a consequence, the eccentricity of at least half of the vertices of the graph are already known. We will explain how to handle the remaining vertices afterwards, in Section~\ref{subsec:peeling}.

\begin{lemma}
Let $G \in \ulog$ and $v_0$ its median vertex. Any vertex $v$ such that $L_{v_0,v} \cap L(G) = \emptyset$ has a weighted eccentricity $\wecc{v}{(G,\omega)} = d(v,v_0) + d_{\omega}(v_0,\umax) = d_{\omega}(v,\umax)$.
\label{le:ecc_large}
\end{lemma}
\begin{proof}
It suffices to show that $v_0 \in I(v,\umax)$. Let $P$ be a $(v,\umax)$-path consisting in the concatenation of a shortest $(v,v_0)$-path (denoted by $Q$) with a shortest $(v_0,\umax)$-path (denoted by $R$). It suffices to show that $P$ is made up of edges belonging to pairwise different $\Theta$-classes (Lemma~\ref{le:signature}). Assume by contradiction that there exists a $\Theta$-class $E_j \in \sigma_{v_0,v} \cap \sigma_{v_0,\umax}$ and that there is no other $\Theta$-class $E_k \in \sigma_{v_0,v} \cap \sigma_{v_0,\umax}$ such that:
\begin{itemize}
\item the edge of $E_k$ on path $Q$ is closer to $v_0$ than the edge of $E_j$ on $Q$,
\item the edge of $E_k$ on path $R$ is closer to $v_0$ than the edge of $E_j$ on $R$.
\end{itemize}
In brief, if we assume $\sigma_{v_0,v} \cap \sigma_{v_0,\umax}$ nonempty, we fix $E_j$ as a $\Theta$-class of this set which is minimal by distance towards the median vertex $v_0$. Such a minimal class necessarily exists.

By definition of $E_j$, there is no $\Theta$-class that appears twice in the path $P_j$, defined as the connected sub-path containing $v_0$ obtained after removing edges of $E_j$ in the path $G[P]$. Hence, $P_j$ is a shortest path. Moreover, its endpoints both belong to the same boundary $\partial H_j''$ of $E_j$. As $\partial H_j''$ is convex (Lemma~\ref{le:boundaries}) and $v_0 \in P_j$, then $v_0 \in \partial H_j''$. This yields a contradiction since $E_j$ is adjacent to $v_0$ and should be part of $L_{v_0,v} \cap L(G)$.
\end{proof}

\begin{definition}[Large sets]
We define recursively a finite sequence $\left(G_i\right)_{0\le i\le \ell}$ of induced subgraphs of $G$ called {\em large sets}. Let $G_0 = G$ and, for any integer $0\le i\le \ell - 1$,
$$G_{i+1} = G_i \setminus H_{i+1}'.$$
In other words, $G_{i+1}$ is incrementally obtained from $G_i$ by removing all vertices in the minority halfspace of $E_{i+1} \in L(G)$.
\label{def:large_sets}
\end{definition}

A first trivial remark is that the number of vertices of the graphs of the sequence is decreasing: $\vert V(G_{i+1}) \vert \le \vert V(G_i) \vert$. A second fact is that each $G_i$ is a gated subgraph of $G$.

\begin{lemma}
For each $0\le i\le \ell$, $G_i$ is a gated subgraph of $G$.
\label{le:gated_large}
\end{lemma}
\begin{proof}
We proceed by induction. Obviously, as the base case, $G_0 = G$ is trivially gated. 

We assume that $G_i$ is a gated subgraph of $G$: $V(G_i)$ is gated in $G$. Moreover, the halfspace $H_{i+1}''$ is also gated in $G$ (Lemma~\ref{le:halfspaces}). As we know from Definition~\ref{def:large_sets}, $V(G_{i+1}) = V(G_i) \setminus H_{i+1}' = V(G_i) \cap H_{i+1}''$. The intersection of two gated sets is gated (Lemma~\ref{le:intersection}). Hence, $V(G_{i+1})$ is gated and our induction step holds.
\end{proof}

In fact, we can rewrite, for any $1\le i\le \ell$, the definition of $G_i$ as the following one:
\begin{equation}
G_i = G[H_1'' \cap H_2'' \cap \ldots \cap H_i''].
\label{eq:large_sets}
\end{equation}

As all large sets $G_i$ are convex/gated subgraphs of $G$, they are median graphs. So, they admit $\Theta$-classes which can in fact be retrieved from the $\Theta$-classes of the original graph $G$.

\begin{lemma}[$\Theta$-classes of isometric subgraphs~\cite{BeDuHa22}]
Let $G$ be a median graph. If $H$ is an isometric\footnote{a subgraph which preserves the distances of the original graph} subgraph of $G$, then the $\Theta$-classes of $H$ are exactly the nonempty subsets among $E_i \cap E(H)$, for $E_i \in \mathcal{E}(G)$.
\label{le:isometric}
\end{lemma}

Since gated subgraphs are also isometric (the converse is false), then the $\Theta$-classes of $G_i$ are inherited from the $\Theta$-classes of $G$. 

We define \textit{slices} as the subgraphs withdrawn from each transition between $G_i$ and $G_{i+1}$.

\begin{definition}[Slices]
We define a finite sequence $(S_i)_{0\le i\le \ell - 1}$ of graphs called {\em slices}:
$$S_i = G_i \setminus V(G_{i+1}).$$
The vertex set of each $G_i$ can be thus partitioned in 2 parts  $V(S_i)$ and $V(G_{i+1})$.
\label{def:corner}
\end{definition}

\begin{figure}[h]
\centering
\scalebox{0.7}{\begin{tikzpicture}

\coordinate (s01) at (1.0,4.6) {};
\coordinate (s02) at (3.2,3.2) {};
\coordinate (s03) at (8.3,3.2) {};
\coordinate (s04) at (8.3,6.3) {};
\coordinate (s05) at (0.2,6.3) {};
\coordinate (s11) at (4.0,3.0) {};
\coordinate (s12) at (4.0,-0.8) {};
\coordinate (s13) at (9.3,-0.8) {};
\coordinate (s14) at (9.3,1.2) {};
\coordinate (s15) at (8.3,1.2) {};
\coordinate (s16) at (8.3,3.0) {};
\coordinate (s17) at (4.0,6.3) {};
\coordinate (s21) at (3.8,-0.8) {};
\coordinate (s22) at (-2.0,-0.8) {};
\coordinate (s23) at (-2.0,6.3) {};
\coordinate (s24) at (-0.1,6.3) {};
\coordinate (s25) at (0.8,4.4) {};
\coordinate (s26) at (3.2,3.0) {};
\coordinate (s27) at (3.8,3.0) {};


\draw[color = white!50!blue, fill = white!90!blue] (s01) -- (s02) -- (s03) -- (s04) -- (s05) -- (s01);
\node[scale=1.4, color = blue] at (7.6,3.7) {$S_0$};

\draw[color = white!50!red, fill = white!90!red] (s11) -- (s12) -- (s13) -- (s14) -- (s15) -- (s16) -- (s11);
\node[scale=1.4, color = red] at (8.7,-0.4) {$S_1$};
\draw[color = red, dashed] (s12) -- (s17);

\draw[color = white!50!gray, fill = white!90!gray] (s21) -- (s22) -- (s23) -- (s24) -- (s25) -- (s26) -- (s27) -- (s21);
\node[scale=1.4, color = black!20!gray] at (-1.0,-0.4) {$V(G_{2})$};

\node[scale=2, rotate=-60] at (-0.2,2) {{\bf \vdots}};
\node[scale=2] at (3,0.2) {{\bf \vdots}};
\node[scale=2, rotate=90] at (5.9,2.5) {{\bf \vdots}};
\node[scale=2, rotate=-20] at (3.4,5.2) {{\bf \vdots}};
\node[scale=2, rotate=45] at (-0.3,5.2) {{\bf \vdots}};


\node[draw, circle, minimum height=0.2cm, minimum width=0.2cm, fill=black] (P11) at (1,1) {};
\node[draw, circle, minimum height=0.2cm, minimum width=0.2cm, fill=black] (P12) at (1,2.5) {};

\node[draw, circle, minimum height=0.2cm, minimum width=0.2cm, fill=black] (P21) at (3,1) {};
\node[draw, circle, minimum height=0.2cm, minimum width=0.2cm, fill=black] (P22) at (3,2.5) {};
\node[draw, circle, minimum height=0.2cm, minimum width=0.2cm, fill=black] (P23) at (3,4) {};

\node[draw, circle, minimum height=0.2cm, minimum width=0.2cm, fill=black] (P31) at (5,1) {};
\node[draw, circle, minimum height=0.2cm, minimum width=0.2cm, fill=black] (P32) at (5,2.5) {};
\node[draw, circle, minimum height=0.2cm, minimum width=0.2cm, fill=black] (P33) at (5,4) {};

\node[draw, circle, minimum height=0.2cm, minimum width=0.2cm, fill=black] (P41) at (1.3,3.8) {};
\node[draw, circle, minimum height=0.2cm, minimum width=0.2cm, fill=black] (P42) at (1.3,5.3) {};
\node[draw, circle, minimum height=0.2cm, minimum width=0.2cm, fill=black] (P43) at (-0.7,3.8) {};


\node[draw, circle, minimum height=0.2cm, minimum width=0.2cm, fill=black] (P52) at (6.8,5.9) {};

\coordinate (i1) at (5.8,4.4) {};
\coordinate (i2) at (6.2,5.3) {};


\node[scale=1.5] at (7.6,5.6) {$\umax$};
\node[scale=1.5] at (2.7,2.2) {$v_0$};
\node[scale=1.5,color = blue] at (0.9,4.5) {$E_1$};
\node[scale=1.5,color = red] at (4.0,1.4) {$E_2$};


\draw[line width = 1.4pt] (P11) -- (P12);
\draw[line width = 1.4pt] (P11) -- (P21);
\draw[line width = 1.4pt] (P12) -- (P22);
\draw[line width = 1.4pt] (P21) -- (P22);

\draw[line width = 1.8pt, color = red] (P21) -- (P31);
\draw[line width = 1.8pt, color = red] (P22) -- (P32);
\draw[line width = 1.4pt] (P31) -- (P32);

\draw[line width = 1.8pt, color = blue] (P22) -- (P23);
\draw[line width = 1.8pt, color = red] (P23) -- (P33);
\draw[line width = 1.8pt, color = blue] (P32) -- (P33);

\draw[line width = 1.4pt] (P22) -- (P41);
\draw[line width = 1.4pt] (P12) -- (P43);
\draw[line width = 1.4pt] (P23) -- (P42);
\draw[line width = 1.4pt] (P41) -- (P43);
\draw[line width = 1.8pt, color = blue] (P41) -- (P42);

\draw[line width = 1.4pt, dashed, in = -120] (P33) to (i1);
\draw[line width = 1.4pt, dashed, in = -130, out = 60] (i1) to (i2);
\draw[line width = 1.4pt, dashed, in = -120] (i2) to (P52);

\end{tikzpicture}}
\caption{Illustration of the slice decomposition of some median graph $G$ with $L(G) = \{E_1,E_2\}$. For example, $V(G_1) = V(S_1) \cup V(G_2)$.}
\label{fig:slices_large}
\end{figure}
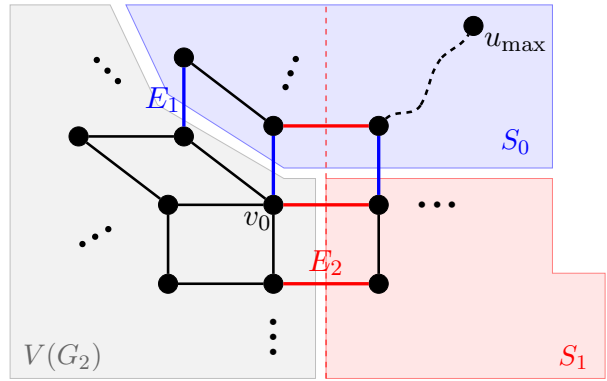

We refer to the expression \textit{slice decomposition} in order to refer to the whole collection of slices. Slices $S_i$ can be rewritten as the intersection of the minority halfspace $H_{i+1}'$ with the current vertex set $V(G_i)$: $S_i = G_i[H_{i+1}']$. Hence, by definition, slices are pairwise vertex-disjoint. In fact,
\begin{equation}
S_i = G\left[H_{i+1}' \setminus \left(H_1'\cup H_2'\cup \ldots H_i'\right)\right].
\label{eq:slices}
\end{equation}

Figure~\ref{fig:slices_large} gives an illustration of the slice decomposition on some example. It represents a part of the graph: only $\umax$ together with the vertices belonging to hypercubes which contain the median vertex $v_0$ are drawn. The wide ladder (ladder set $L_{v_0,\umax}$) contains two $\Theta$-classes $E_1$ and $E_2$: $\ell = 2$. In such case, there are two slices: $S_0$ is exactly the minority halfspace $H_1'$ and $S_1$ is $H_2'$ deprived of the vertices of $H_1'$: $S_1 = G[H_2' \setminus H_1']$. Sets $S_0$, $S_1$, and $V(G_{2})$ are disjoint and cover the graph.

\begin{lemma}
For each $0\le i\le \ell-1$, $S_i$ is a gated subgraph of $G$
\label{le:gated_slices}
\end{lemma}
\begin{proof}
Graph $G_i$ is gated subgraph of $G$ (Lemma~\ref{le:gated_large}). Consequently, $S_i$ is gated since it is the intersection between $H_{i+1}'$ and $V(G_i)$ (Lemma~\ref{le:intersection}).
\end{proof}

In terms of cardinality, slices are thus small, $\vert V(S_i) \vert \le \vert H_{i+1}' \vert \le \frac{n}{2\log(n)}$, while even the last large set $G_{\ell}$ contains at least half of the vertices.

\begin{lemma}[Cardinality of large sets]
For any $0\le i\le \ell$, $\vert V(G_i) \vert \ge n(1-\frac{i}{2\log(n)})$.
\label{le:card_slices}
\end{lemma}
\begin{proof}
By induction: $i=0$ trivially holds. The large set $G_{i+1}$ corresponds to $G_i$ after withdrawing the slice $S_i$, which contains at most $\frac{n}{2\log n}$ vertices. Therefore by induction, $$\vert V(G_{i+1}) \vert \ge \vert V(G_i)\vert - \frac{n}{2\log n} = n\left(1-\frac{i}{2\log(n)}\right) - \frac{n}{2\log n} = n\left(1-\frac{i+1}{2\log(n)}\right).$$
Which yields the announced inequality for all integers $0\le i\le \ell-1$.
\end{proof}

Indeed, as $\ell \le d\le \log n$, we have $\vert V(G_{\ell})\vert \ge  n(1-\frac{\ell}{2\log(n)}) \ge \frac{n}{2}$. From Equation~\eqref{eq:large_sets}, set $V(G_{\ell})$ contains exactly the vertices which does not belong to any $H_1',\ldots H_{\ell}'$, in other words whose ladder set has no intersection with $L(G)$. As stated in Lemma~\ref{le:ecc_large}, the weighted eccentricities of all vertices $v$ of $V(G_{\ell})$ can be directly deduced. On the other hand, any vertex $v$ with a ladder set $L_{v_0,v}$ intersecting $L(G)$ belongs to one of the slices $S_0,\ldots,S_{\ell-1}$ (Equation~\eqref{eq:slices}). The input graph, by definition of slices and large sets, satisfies $V(G) = V(G_{\ell}) \cup V(S_0) \cup V(S_1) \cup \ldots \cup V(S_{\ell-1})$. Furthermore, the large sets generally satisfy, for any $0\le i\le \ell-1$: $$V(G_i) = V(G_{\ell}) \cup V(S_i) \cup V(S_{i+1}) \cup \ldots \cup V(S_{\ell-1}).$$

\subsection{Peeling the slices} \label{subsec:peeling}

We show that, assuming we know all the weighted eccentricities on slices, which are small-sized induced subgraphs of the input graph $(G,\omega) \in \ulog$, we can deduce all weighted eccentricities of $G$. We begin with a description of how to modify weights on slices so that our recursive calls will enable us to retrieve all eccentricities of $(G,\omega)$.

\begin{definition}
Let $(G,\omega) \in \ulog$ with $\umax \neq v_0$ and an integer $0\le i\le \ell-1$. The weight function $\omega_i^*$ on slice $S_i$ is defined as:
\begin{itemize}
\item for any $x \in \partial H_{i+1}' \cap V(S_i)$, $\omega_i^*(x)$ is equal to 
the maximum distance $d_{\omega}(x,z)$ where $z$ belongs to the fiber $F_{H_{i+1}'}(x)$ of $G_i$.\footnote{Here, we abuse notation: $H_{i+1}'$ refers to the intersection of $V(G_i)$ with $H_{i+1}'$, which is a halfspace of the $\Theta$-class $E_{i+1} \cap E(G_i)$ (Lemma~\ref{le:isometric}).}
\item for any $y \in V(S_i) \setminus \partial H_{i+1}'$, $\omega_i^*(y) = \omega(y)$.
\end{itemize}
\label{def:weights_slices}
\end{definition}

We state now a theorem which introduces our recursive process to compute the weighted eccentricities of $(G,\omega) \in \ulog$. As we already observed it, all weighted eccentricities of the large set $G_{\ell}$ are already known. Therefore, we focus on computing the eccentricities of the vertices which belong to at least one slice. 

Let us now define a labeling function $\mathcal{B}_i$, for any $0\le i\le \ell-1$ on the vertices of slices $S_i \cup S_{i+1} \cup \ldots \cup S_{\ell-1}$. The label of $\mathcal{B}_i$ for some vertex $u$ is denoted by $\mathcal{B}_i(u)$. 

\begin{definition}[Labeling $\mathcal{B}_i$] For any $0\le i\le \ell-1$ and any vertex $u \in V(S_i) \cup V(S_{i+1}) \cup \ldots \cup V(S_{\ell-1})$, label $\mathcal{B}_i(u)$ is equal to the weighted eccentricity of $u$ on graph $(G_i,\omega)$:
$$\mathcal{B}_i(u) = \wecc{x}{(G_i,\omega)}.$$
\label{def:labeling_beta}
\end{definition}

Observe that labeling $\mathcal{B}_0$ would allow us to determine all eccentricities of $(G,\omega)$: for vertices in the slices, the labels give their eccentricity while for vertices in $G_{\ell}$, we can use Lemma~\ref{le:ecc_large}. 

Our proposition is the following. Assume that we already computed some labeling $\mathcal{B}_{i+1}$ and our intention is to obtain $\mathcal{B}_i$. Our idea to achieve this recursive step consists in launching two recursive calls, in order to compute the weighted eccentricities on graph $S_i$ but with two different weight functions $\omega$ and $\omega_i^*$. We prove that given labeling $\mathcal{B}_{i+1}$ and the results of these calls, we can retrieve labeling $\mathcal{B}_{i}$.

\begin{theorem}
Let $0\le i\le \ell-1$. Assume that the following values are known:
\begin{itemize}
\item all weighted eccentricities of $(S_i,\omega)$,
\item all weighted eccentricities of $(S_i,\omega_i^*)$,
\item all labels of $\mathcal{B}_{i+1}$.
\end{itemize}
One can compute in linear time $O(\vert E(G_i)\vert)$ all labels of $\mathcal{B}_{i}$.
\label{th:unbalanced_recursive}
\end{theorem}
\begin{proof}
We apply the BFS traversal evoked in Lemma~\ref{le:compute_gates} for graph $G_i$ with the gated set $S_i$. We recall that the vertex set of $G_i$ can be partitioned into $V(G_{i+1})$ and $V(S_i)$. We thus obtain in time $O(\vert E(G_i)\vert)$ all the unweighted distances $d(g_{S_i}(v),v)$, for each $v \in V(G_{i+1})$, and as a consequence, also all weighted distances $d_{\omega}(g_{S_i}(v),v) = d(g_{S_i}(v),v) + \omega(v)$. Our objective is to compute the labels $\mathcal{B}_i(x)$. We distinguish two cases: either $x$ belongs to the current slice $S_i$, or in a former slice $S_{i+1} \cup \ldots \cup S_{\ell-1}$.

\textit{Case 1: $x \in V(S_i)$}. We determine the labels $\mathcal{B}_i(x)$ for vertices $x \in S_i$ which are from Definition~\ref{def:labeling_beta} their eccentricity in $(G_i,\omega)$. The weighted distance between some $x \in S_i$ and $v \in V(G_{i+1})$ is:
\begin{equation}d_{\omega}(x,v) = d(x,g_{S_i}(v)) + d_{\omega}(g_{S_i}(v),v) = d(x,g_{S_i}(v)) + d(g_{S_i}(v),v) + \omega(v).
\label{eq:one_step_unbalanced}
\end{equation}
Let us recall that, from Definition~\ref{def:weights_slices}, when $x$ belongs to the boundary of $S_i$, $\omega_i^*(x)$ is equal to the maximum $d_{\omega}(x,v)$ such that $v \in F_{S_i}(x)$ - or, said differently, $x = g_{S_i}(v)$. Otherwise, $\omega_i^*(x) = \omega(x)$.

We consider two cases for vertex $x \in V(S_i)$. If $x \in \partial H_{i+1}'$, its weighted eccentricity is achieved:
\begin{itemize}
    \item either with a vertex $z$ of $S_i$: $d_{\omega}(x,z) = \wecc{x}{(S_i,\omega)}$, known by the assumption on $(S_i,\omega)$,
    \item or with a vertex $z$ of $F_{S_i}(x)$: $d_{\omega}(x,z) = \omega_i^*(x)$, computed with the BFS traversal,
    \item or with a vertex $z$ of $V(G_{i+1})$ not in the fiber $F_{S_i}(x)$: in this case, it will be given by $d_{\omega}(x,z) = d(x,g_{S_i}(z)) + \omega_i^*(g_{S_i}(z))$.
\end{itemize}
If the weighted eccentricity of $x$ is attained for some vertex $z \in V(G_{i+1})$ (which corresponds to the two latter bullets), its value is equal to $\wecc{x}{(S_i,\omega_i^*)}$. Written briefly, 
\begin{equation}
    \wecc{x}{(G_i,\omega)} = \max \left\{\wecc{x}{\left(S_i,\omega\right)}, ~ \wecc{x}{(S_i,\omega_i^*)}\right\}.
    \label{eq:wecc_boundary}
\end{equation} 
From the statement assumptions, all the values needed to compute $\wecc{x}{(G_i,\omega)}$ in Equation~\eqref{eq:wecc_boundary} are supposed to be known.

If $x \in V(S_i) \setminus \partial H_{i+1}'$, the reasoning is relatively similar: either the weighted eccentricity of $x$ is achieved either with some $v \in S_i$, or with $v \in V(G_{i+1})$. In the latter case, from Equation~\eqref{eq:one_step_unbalanced}, the weighted eccentricity of $x$ is $d(x,g_{S_i}(v)) + \omega_i^*(g_{S_i}(v)) = d_{\omega_i^*}(x,g_{S_i}(v))$. As the weight difference between $\omega_i^*$ and $\omega$ lies only on the boundary, we have: $\wecc{x}{(G_i,\omega)} = \max \left\{\wecc{x}{\left(S_i,\omega\right)},~ \wecc{x}{(S_i,\omega_i^*)}\right\}$. Figure~\ref{fig:peeling} illustrates these two possibilities on the example already introduced in Figure~\ref{fig:slices_large}: either the maximum weighted distance from $x$ goes towards $S_i$ or towards $V(G_{i+1})$.  We completed the first part of the proof: the weighted eccentricities of all $x \in S_i$ in graph $(G_i,\omega)$ can be directly deduced from the theorem assumptions: $\mathcal{B}_i(x) = \wecc{x}{(G_i,\omega)}$.

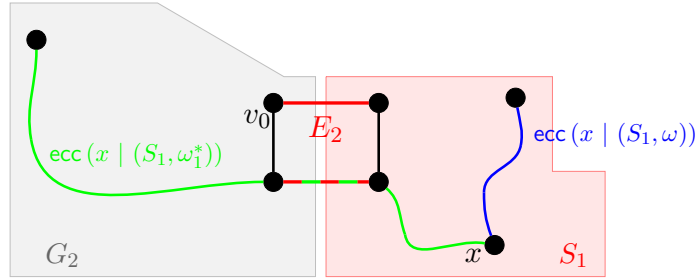
\begin{figure}[h]
\centering
\scalebox{0.7}{\begin{tikzpicture}

\coordinate (s11) at (4.0,3.0) {};
\coordinate (s12) at (4.0,-0.8) {};
\coordinate (s13) at (9.3,-0.8) {};
\coordinate (s14) at (9.3,1.2) {};
\coordinate (s15) at (8.3,1.2) {};
\coordinate (s16) at (8.3,3.0) {};
\coordinate (s17) at (4.0,6.3) {};
\coordinate (s21) at (3.8,-0.8) {};
\coordinate (s22) at (-2.0,-0.8) {};
\coordinate (s23) at (-2.0,4.4) {};
\coordinate (s25) at (0.8,4.4) {};
\coordinate (s26) at (3.2,3.0) {};
\coordinate (s27) at (3.8,3.0) {};
\coordinate (i1) at (7.0,0.8) {};
\coordinate (i2) at (7.7,1.7) {};
\coordinate (i3) at (5.5,0) {};
\coordinate (i4) at (-1.6,1.8) {};


\draw[color = white!50!red, fill = white!90!red] (s11) -- (s12) -- (s13) -- (s14) -- (s15) -- (s16) -- (s11);
\node[scale=1.4, color = red] at (8.7,-0.4) {$S_1$};;

\draw[color = white!50!gray, fill = white!90!gray] (s21) -- (s22) -- (s23) -- (s25) -- (s26) -- (s27) -- (s21);
\node[scale=1.4, color = black!20!gray] at (-1.0,-0.4) {$G_{2}$};


\node[draw, circle, minimum height=0.2cm, minimum width=0.2cm, fill=black] (P21) at (3,1) {};
\node[draw, circle, minimum height=0.2cm, minimum width=0.2cm, fill=black] (P22) at (3,2.5) {};

\node[draw, circle, minimum height=0.2cm, minimum width=0.2cm, fill=black] (P31) at (5,1) {};
\node[draw, circle, minimum height=0.2cm, minimum width=0.2cm, fill=black] (P32) at (5,2.5) {};

\node[draw, circle, minimum height=0.2cm, minimum width=0.2cm, fill=black] (P4) at (7.2,-0.2) {};
\node[draw, circle, minimum height=0.2cm, minimum width=0.2cm, fill=black] (P5) at (7.6,2.6) {};
\node[draw, circle, minimum height=0.2cm, minimum width=0.2cm, fill=black] (P6) at (-1.5,3.7) {};


\node[scale=1.5] at (2.7,2.2) {$v_0$};
\node[scale=1.5] at (6.8,-0.4) {$x$};
\node[scale=1.5,color = red] at (4.0,2.0) {$E_2$};
\node[scale = 1.2, color = blue] at (9.5,1.9) {$\wecc{x}{(S_1,\omega)}$};
\node[scale = 1.2, color = green] at (0.4,1.5) {$\wecc{x}{(S_1,\omega_1^*)}$};


\draw[line width = 1.8pt, bicolor={red}{green}]
  (P21) to (P31);

\draw[line width = 1.4pt] (P21) -- (P22);
\draw[line width = 1.8pt, color = red] (P22) -- (P32);
\draw[line width = 1.4pt] (P31) -- (P32);

\draw[line width = 1.4pt, color = blue, out = 110, in = -90] (P4) to (i1);
\draw[line width = 1.4pt, color = blue, out = 90, in = -100] (i1) to (i2);
\draw[line width = 1.4pt, color = blue, out = 80, in = -80] (i2) to (P5);

\draw[line width = 1.4pt, color = green, out = 170, in = -80] (P4) to (i3);
\draw[line width = 1.4pt, color = green, out = 110, in = -40] (i3) to (P31);
\draw[line width = 1.4pt, color = green, out = -180, in = -80] (P21) to (i4);
\draw[line width = 1.4pt, color = green, out = 100, in = -90] (i4) to (P6);

\end{tikzpicture}}
\caption{Retrieving the eccentricities of $x \in V(S_i)$ in graph $G_i$: example with $\ell = 2$ and $i = 1$. Either the largest weighted distance from $x$ is with a vertex of $S_i$, or with one in the other side $V(G_{i+1})$.}
\label{fig:peeling}
\end{figure}

\textit{Case 2: $x \in V(S_{i+1}) \cup \ldots \cup V(S_{\ell-1})$}. We determine the labels $\mathcal{B}_i(x)$ for vertices $x \in V(S_{i+1}) \cup \ldots \cup V(S_{\ell-1})$. Since we know all labels of $\mathcal{B}_{i+1}$, we have the eccentricity $\mathcal{B}_{i+1}(x) = \wecc{x}{(G_{i+1},\omega)}$ of $x$ in graph $(G_{i+1},\omega)$. By definition, $\mathcal{B}_i(x) \ge \mathcal{B}_{i+1}(x)$ and the only chance to have $\mathcal{B}_i(x) > \mathcal{B}_{i+1}(x)$ is that the weighted eccentricity of $x$ in $G_i$ is the distance from $x$ to some vertex $v \in V(S_i)$, as $V(G_i)\setminus V(G_{i+1}) = V(S_i)$. So, we compute the maximum distance $d_{\omega}(x,v)$ for $v \in V(S_i)$ and verify whether it is larger than $\mathcal{B}_{i+1}(x)$ or not. Thanks to the BFS traversal, we know $g_{S_i}(x)$ and then the maximum distance $D(x) = \max \{d_{\omega}(x,v) : v\in V(S_i)\}$ is given:
$$D(x) = d(x,g_{S_i}(x)) + \wecc{g_{S_i}(x)}{(S_i,\omega)}.$$
In summary, all these values can be computed since not only $d(x,g_{S_i}(x))$ but also all eccentricities of $(S_i,\omega)$ are known: 
$$\mathcal{B}_i(x) = \max \left\{\mathcal{B}_{i+1}(x),~d(x,g_{S_i}(x)) + \wecc{g_{S_i}(x)}{(S_i,\omega)} \right\}$$.
\end{proof}

The latter theorem stands as the recursive step to prove that the weighted eccentricities of $(G,\omega) \in \ulog$ can be computed in linear time, assuming that we already have the weighted eccentricities of all slices. 

\begin{corollary}
    Let $(G,\omega) \in \ulog$ with $\umax \neq v_0$ given with its slice decomposition. Assume that, for each integer $0\le i\le \ell-1$, all weighted eccentricities of $(S_i,\omega)$ but also all weighted eccentricities of $(S_i,\omega_i^*)$ are known. Then, one can compute all weighted eccentricities of $(G,\omega)$ in time $O(n\log^2 n)$.
    \label{co:wecc_slice}
\end{corollary}
\begin{proof}
We proceed by induction in order to prove that labeling $\mathcal{B}_i$ can be determined in time $O((\ell-i)m)$. The base case is the computation of labeling $\mathcal{B}_{\ell-1}$. We focus on graph $G_{\ell-1}$ which is partitioned by $\Theta$-class $E_{\ell}$ into two halfspaces: $S_{\ell-1}$ and $V(G_{\ell})$. The idea is similar to the ones explained in the proof of Theorem~\ref{th:unbalanced_recursive}. We begin by applying the BFS traversal of Lemma~\ref{le:compute_gates} in order to retrieve the $S_{\ell-1}$-gates of any vertex of graph $G_{\ell-1}$. For $x\in S_{\ell - 1}$, its weighted eccentricity in $G_{\ell-1}$ is given by: 
$$\mathcal{B}_{\ell-1}(x) = \max \{ \wecc{x}{\left(S_{\ell-1},\omega\right)},~ \wecc{x}{(S_{\ell-1},\omega_{\ell-1}^*)}\}.$$
The knowledge of the eccentricities of $(S_{\ell-1},\omega)$ and $(S_{\ell-1},\omega_{\ell-1}^*)$, but also of all distances to $S_{\ell-1}$-gates in $G_{\ell-1}$ thanks to the BFS, allows us to obtain the labeling $\mathcal{B}_{\ell-1}$. The running time is made up of the BFS traversal with a linear number of maximum computations, which is $O(m)$ and independent from $\ell$.

The inductive step consists in applying Theorem~\ref{th:unbalanced_recursive}. Indeed, labeling $\mathcal{B}_{i}$ can be determined from labeling $\mathcal{B}_{i+1}$ in time $O(m)$, given that we have the eccentricities of both $(S_i,\omega)$ and $(S_i,\omega_i^*)$. Hence, from the induction hypothesis, the running time of the whole process is $O((\ell-i-1)m + m) = O((\ell-i)m)$.

Consequently, labeling $\mathcal{B}_0$ is obtained in time $O(dm)=O(n\log^2 n)$, as $\ell \le d \le \log_2 n$. For any vertex $x$ belonging to a slice, {\em i.e.} $x \in S_0 \cup \ldots \cup S_{\ell-1}$, value $\mathcal{B}_0(x)$ is its weighted eccentricity on $(G_0,\omega) = (G,\omega)$. Concerning vertices $z \in V(G_{\ell})$, as observed earlier, their ladder set $L_{v_0,z}$ has no intersection with $L(G)$, so their weighted eccentricity can be directly computed (Lemma~\ref{le:ecc_large}). All in all, the weighted eccentricities of $(G,\omega)$ are known.
\end{proof}

\subsection{Summary and analysis for the eccentricities computation}

Given the observations and subroutines presented above, we present now our recursive algorithm which computes all weighted eccentricities of a median graph $(G,\omega)$. We call this algorithm \algo\ (acronym for Median $\Theta$-class Recursive Scheme for Eccentricities) and its pseudocode is given in Algorithm~\ref{alg:wecc}. The base case of \algo\ is when the number of vertices of the graph is at most $2$. Trivially, one can compute the weighted eccentricites of such graph in constant time $O(1)$ (line~\ref{line:base} of Algorithm~\ref{alg:wecc}). We focus now on the recursive calls launched by \algo .

It starts with the computation of $\Theta$-classes and their halfspaces sizes (line~\ref{line:classes}) which can be done in time $O(n\log n)$, according to Lemmas~\ref{le:linear_classes} and~\ref{le:halfspaces}. Then, it distinguishes two cases: either the input graph $(G,\omega)$ contains at least one $f$-balanced $\Theta$-class (with $f = 2\log$) or not.

If $(G,\omega)$ admits a balanced $\Theta$-class $E_i$, then the procedure is relatively simple. We call recursively \algo\ in order to obtain the weighted eccentricities of both halfspaces, {\em i.e.} of $G[H_i']$ and $G[H_i'']$ (lines~\ref{line:call1}-\ref{line:call2}). Thanks to Theorem~\ref{th:balanced_recursion}, we know that we can retrieve all eccentricities of $(G,\omega)$ in linear time (line~\ref{line:retrieve}).

\begin{algorithm}[t]
    \SetKwFor{For}{for}{do}{\nl endfor}
    \SetKwFor{Forall}{for all}{do}{\nl endfor}
    \SetKwIF{If}{ElseIf}{Else}{if}{then}{else if}{else}{}
    \DontPrintSemicolon
    \SetNlSty{}{}{:}
    \SetAlgoNlRelativeSize{0}
    \SetNlSkip{1em}

 	\nl \KwIn{A weighted median graph $(G,\omega)$.}
 	\nl \KwOut{All eccentricities $\wecc{u}{(G,\omega)}$ for each $u \in V(G)$.}
	
 	\nl \If{$\vert V(G) \vert \le 2$}{
 	    \nl \Return{} all eccentricities by enumerating all (at most $2$) shortest paths\label{line:base}\;
 	}
 	\nl \ifend\;
 	\nl Compute all $\Theta$-classes $E_i \in \mathcal{E}(G)$ (Lemma~\ref{le:linear_classes}) and their halfspaces sizes (Lemma~\ref{le:halfspaces}) \label{line:classes}\;
 	\nl \If{there exists a $\Theta$-class $E_i$ which is balanced}{
 	    \nl {\sc list1} $\leftarrow \algo(G[H_i'],\omega)$ \label{line:call1}\;
 	    \nl {\sc list2} $\leftarrow \algo(G[H_i''],\omega)$ \label{line:call2}\;
 	    \nl \Return{} all eccentricities of $(G,\omega)$ retrieved from {\sc list1}/{\sc list2} (Theorem~\ref{th:balanced_recursion}) \label{line:retrieve}\;
 	} \Else {
 	    \nl Compute the unique median vertex $v_0$ as $G \in \ulog$ (Lemma~\ref{le:unique_median})\label{line:median}\;
 	    \nl Launch a BFS starting from $v_0$ to determine all distances from $v_0$, but also $\umax$ and $L(G)$\label{line:bfs}\;
 	    \nl \lIf{$\umax = v_0$}{\Return{} all eccentricities retrieved with Lemma~\ref{le:v0farthest}}
 	    \nl Give a unique (arbitrary) identifier for each class of $L(G)$, from $1$ to $\ell = \vert L(G) \vert$\label{line:reorder}\;
 	    \nl Determine all large sets $G_j$ and slices $S_i$ for each $0\le j\le \ell$ and $0\le i\le \ell - 1$\;
 	    \nl {\sc collec} $\leftarrow$ list of $\ell$ empty lists; {\sc wcollec} $\leftarrow$ list of $\ell$ empty lists;\;
 	    \nl \For{$i$ from 0 to $\ell-1$}{
 	        \nl {\sc collec}[$i$] $\leftarrow \algo(S_i,\omega)$\;
 	        \nl {\sc wcollec}[$i$] $\leftarrow \algo(S_i,\omega_i^*)$\;
 	    }
 	    \nl \Return{} all eccentricities of $(G,\omega)$ retrieved from {\sc collec} and {\sc wcollec} (Theorem~\ref{th:unbalanced_recursive})\; 
 	}
 	\nl \ifend\;
	
 	\caption{Algorithm \algo}
 	\label{alg:wecc}
 \end{algorithm}

Otherwise, $(G,\omega) \in \ulog$, and there is a unique median vertex $v_0$ which belong to all majoritarian halfspaces of $G$ (Lemma~\ref{le:unique_median}). It is computed in linear time (Corollary~\ref{co:median_linear}). Then, we execute a BFS from $v_0$ (line~\ref{line:bfs}) which allows us to determine $\umax$ and $L(G)$, according to Lemma~\ref{le:ladder_linear}. The slice decomposition of $G$ can be directly deduced from this BFS since the vertices of $S_i$ are exactly the vertices $v$ such that $E_{i+1} \in L_{v_0,v}$ but $E_1,\ldots,E_i \notin L_{v_0,v}$, see Equation~\eqref{eq:slices}. Then, we apply recursively \algo\ on all instances $(S_i,\omega)$ and $(S_i,\omega_i^*)$, with $0\le i\le \ell -1$. Eventually, one can retrieve the eccentricites of $(G,\omega)$ in linear time according to Theorem~\ref{th:unbalanced_recursive}.

We show that \algo\ is executed in quasilinear time.

\mainEcc*
\begin{proof}
Both Theorems~\ref{th:balanced_recursion} and~\ref{th:unbalanced_recursive} ensure us that \algo\ computes exactly all eccentricities of the weighted median graph $(G,\omega)$. Our effort consists now in showing that its running time is quasilinear.

We prove the following statement: if $\vert V(G) \vert \ge 3$, the computation of $\algo(G,\omega)$, in addition with a running time $O(\vert V(G)\vert\log^2 \vert V(G)\vert)$, launches recursive calls on a collection $\mathcal{C}$ of weighted subgraphs of $G$ satisfying:
\begin{itemize}
    \item $\sum\limits_{(G',\omega') \in \mathcal{C}} \vert V(G') \vert \le \vert V(G) \vert$
    \item $\max \left\{\vert V(G') \vert : (G',\omega') \in \mathcal{C} \right\} \le \vert V(G) \vert\left(1-\frac{1}{2\log(\vert V(G) \vert)}\right)$
\end{itemize}

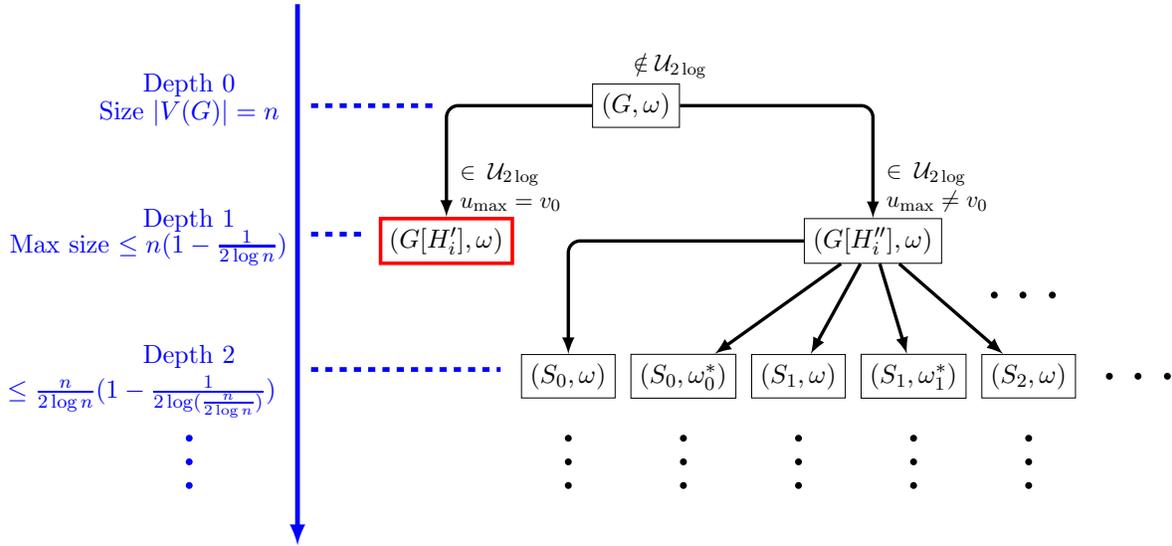
\begin{figure}[h]
\centering
\scalebox{0.9}{\begin{tikzpicture}


\node[draw] (P11) at (7.0,10.5) {$(G,\omega)$};

\node[draw = red, line width = 1.5pt, minimum height=0.2cm, minimum width=0.2cm] (P21) at (4.2,8.5) {$(G[H_i'],\omega)$};
\node[draw, minimum height=0.2cm, minimum width=0.2cm] (P22) at (10.5,8.5) {$(G[H_i''],\omega)$};

\node[draw, minimum height=0.2cm, minimum width=0.2cm] (P31) at (6.0,6.5) {$(S_0,\omega)$};
\node[draw, minimum height=0.2cm, minimum width=0.2cm] (P32) at (7.7,6.5) {$(S_0,\omega_0^*)$};
\node[draw, minimum height=0.2cm, minimum width=0.2cm] (P33) at (9.4,6.5) {$(S_1,\omega)$};
\node[draw, minimum height=0.2cm, minimum width=0.2cm] (P34) at (11.1,6.5) {$(S_1,\omega_1^*)$};
\node[draw, minimum height=0.2cm, minimum width=0.2cm] (P35) at (12.8,6.5) {$(S_2,\omega)$};


\node[color = blue] at (0.4,10.8) {Depth $0$};
\node[color = blue] at (0.4,10.4) {Size $\vert V(G) \vert = n$};
\node[color = blue] at (0.4,8.8) {Depth $1$};
\node[color = blue] at (-0.2,8.4) {Max size $\le n(1-\frac{1}{2\log n})$};
\node[color = blue] at (0.4,6.8) {Depth $2$};
\node[color = blue] at (-0.3,6.2) {$\le \frac{n}{2\log n}(1-\frac{1}{2 \log(\frac{n}{2\log n})})$};

\node[scale=2.5] at (12.8,7.7) {\ldots};
\node[scale=2.5] at (14.5,6.5) {\ldots};

\node[scale=2.5] at (6.0,5.5) {\vdots};
\node[scale=2.5] at (7.7,5.5) {\vdots};
\node[scale=2.5] at (9.4,5.5) {\vdots};
\node[scale=2.5] at (11.1,5.5) {\vdots};
\node[scale=2.5] at (12.8,5.5) {\vdots};
\node[scale=2.5, color = blue] at (0.4,5.5) {\vdots};

\node[scale = 0.9] at (7.5,11.1) {$\notin \ulog$};
\node[scale = 0.9, text width = 2cm] at (5.3,9.3) {$\in \ulog$ $\umax = v_0$};
\node[scale = 0.9, text width = 2cm] at (11.6,9.3) {$\in \ulog$ $\umax \neq v_0$};

\draw[line width = 2pt, color=blue, dashed] (2.2,10.5)--(4.0,10.5);
\draw[line width = 2pt, color=blue, dashed] (2.2,8.6)--(3.0,8.6);
\draw[line width = 2pt, color=blue, dashed] (2.2,6.6)--(5.0,6.6);
\draw[->,>=latex,line width = 2pt, color=blue] (2.0,12.0)--(2.0,4.0);

\draw[->,>=latex,rounded corners=5pt,line width = 1.4pt] (P11) -| (P21);
\draw[->,>=latex,rounded corners=5pt,line width = 1.4pt] (P11) -| (P22);

\draw[->,>=latex,rounded corners=5pt,line width = 1.4pt] (P22) -| (P31);
\draw[->,>=latex,line width = 1.4pt] (P22) -- (P32);
\draw[->,>=latex,line width = 1.4pt] (P22) -- (P33);
\draw[->,>=latex,line width = 1.4pt] (P22) -- (P34);
\draw[->,>=latex,line width = 1.4pt] (P22) -- (P35);

\end{tikzpicture}}
\caption{Tree of recursive calls: an example where the input $(G,\omega)$ admits a balanced $\Theta$-class $E_i$; then $(G[H_i'],\omega) \in \ulog$ with $\umax = v_0$ hence its weighted eccentricities can be directly computed in linear time (Corollary~\ref{co:casev0}); finally $(G[H_i''],\omega) \in \ulog$ but $\umax \neq v_0$ so $2\ell$ recursive calls are launched.}
\label{fig:recurrence}
\end{figure}

If $(G,\omega)$ admits a balanced $\Theta$-class $E_i$, two recursive calls are achieved on its halfspaces which partition the vertex set of $G$. Moreover, from Definition~\ref{def:balanced}, none of these halfspaces has a size greater than $n\left(1-\frac{1}{2\log n}\right)$ with $n = \vert V(G) \vert$. Otherwise, when $(G,\omega) \in \ulog$, the situation is trickier. If $\umax = v_0$, all eccentricities can be computed in linear time (Corollary~\ref{co:casev0}) and there is no recursive call (red node of the tree in Figure~\ref{fig:recurrence}). But, if $\umax \neq v_0$, two recursive calls are launched for each slice $S_i$, $0\le i\le \ell - 1$: one with weight function $\omega$ and one with $\omega_i^*$. Each slice is the subset of a minority halfspace, hence $\vert S_i \vert \le \frac{n}{2\log n} < n\left(1-\frac{1}{2\log n}\right)$. Then, observe that two different slices do not intersect each other by definition. As $\ell \le \log (\vert V(G) \vert)$, we have:
$$\sum\limits_{(G',\omega') \in \mathcal{C}} \vert V(G') \vert = \sum\limits_{i=0}^{\ell-1} 2\vert S_i \vert \le \frac{2\ell \vert V(G) \vert}{2\log(\vert V(G) \vert)} \le \vert V(G) \vert.$$
In both cases, the inequalities claimed above are satisfied. We analyze the tree of recursive calls, admitting the input $(G,\omega)$ as a root and, as leaves, either cases $\umax = v_0$ or very small weighted graphs (at most two vertices). At each depth $k$ of the tree, the total number of vertices involved in the instances is at most $n = \vert V(G) \vert$. The extra running time needed to retrieve the weighted eccentricities of the instances at depth $k$ from the weighted eccentricities of instances at depth $k+1$ is $O(n\log^2 n)$: it is $O(n\log n)$ when a balanced $\Theta$-class is present (Theorem~\ref{th:balanced_recursion}), $O(n\log^2 n)$ otherwise (Corollary~\ref{co:wecc_slice}). Let $s_{k}$ be the maximum number of vertices of some weighted graph at depth $k \ge 0$. Obviously, $s_0 = n$. Sequence $(s_k)$ follows the scheme introduced in Lemma~\ref{le:log_decrease} with $\lambda = 2$: in particular, when $s_k \ge 3$, then $s_{k+1} \le \lfloor s_k(1-\frac{1}{2\log(s_k)}) \rfloor$. Indeed, according to the second inequality above, as $s_k$ is the maximum size of some instance at depth $k$, any instance at depth $k+1$ cannot exceed $\lfloor s_k(1-\frac{1}{2\log(s_k)}) \rfloor$. As a conclusion, the depth of the recursive tree cannot overpass $2\log^2 n$ (Lemma~\ref{le:log_decrease}).

As the running time needed per depth is $O(n\log^2 n)$ and the depth of the tree of recursive calls is upper-bounded by $2\log^2 n$, then the total running time of \algo\ is $O(n\log^4 n)$.
\end{proof}

\section{Distance oracle} \label{sec:do}

We present in this section the second outcome of this paper, which is the design of a distance oracle (DO) for (unweighted) median graphs with a poly-logarithmic size of labels and query time. In this section, the notions of balanced/unbalanced $\Theta$-classes will be associated with a different function $f$ than in Sections~\ref{sec:first_step} and~\ref{sec:unbalanced}. Indeed, we fix $f$ as the constant function: $f : n \rightarrow 3$. We denote by $\ut$ the set of median graphs without any $3$-balanced $\Theta$-classes. In other words, if $G \in \ut$, then for each $E_i \in \mathcal{E}(G)$, either $\vert H_i' \vert < \frac{n}{3}$ or $\vert H_i'' \vert < \frac{n}{3}$. 

The main technique used to achieve this goal is similar to the one used in Sections~\ref{sec:first_step} and~\ref{sec:unbalanced}: we exploit balanced $\Theta$-classes. We propose a recursive scheme where, at each step, a non-negligible number of vertices is withdrawn for each recursive call. When the number of vertices $n = \vert V(G) \vert$ is at most $2$, the label of each vertex has size $0$: it does not contain any bit. Indeed, the distance towards the other vertex is necessarily $1$, as $G$ is connected.

Now we fixed this trivial case, we focus on median graphs $G$ with $n \ge 3$. We distinguish two cases: either $G$ admits a $3$-balanced $\Theta$-class, or $G \in \ut$.

\subsection{Recursive scheme exploiting gated halfspaces}

We propose a distance oracle $\Lambda_G$ for median graphs $G$. Any label $\Lambda_G(u)$ is a sequence of bits. The size of the DO is the maximum size of all sequences $\Lambda_G(u)$: we denote $\vert \Lambda_G \vert = \max_{u\in V(G)} \vert \Lambda_G(u) \vert$. Then, we say that the DO $\Lambda_G$ has a query time $\tau(n)$ if, for any pair $u,v \in V(G)$, the time needed to retrieve $d(u,v)$ thanks to labeling $\Lambda_G$ is at most $\tau(n)$. The DO we propose has both poly-logarithmic size and query time. 

\begin{theorem}
    Let $G$ be a median graph and $E_i \in \mathcal{E}(G)$. Assume that:
    \begin{itemize}
        \item a DO $\Lambda_{H_i'}$ of graph $G[H_i']$ is known with query time $\tau_1(\vert H_i' \vert)$,
        \item a DO $\Lambda_{H_i''}$ of graph $G[H_i'']$ is known with query time $\tau_2(\vert H_i'' \vert)$.
    \end{itemize}
    Then, one can build a DO $\Lambda_G$ of graph $G$ with size at most $3\log_2 n + 1 + \max \left\{\vert \Lambda_{H_i'} \vert, \vert \Lambda_{H_i''} \vert\right\}$ and query time $\tau(n) = O(\log_2 n) + \max \left\{\tau_1(\vert H_i' \vert), \tau_2(\vert H_i'' \vert) \right\}$. The construction takes time $O(n(\log_2 n + \max \{\vert \Lambda_{H_i'} \vert, \vert \Lambda_{H_i''} \vert\}))$.
\label{th:balanced_do}
\end{theorem}
\begin{proof}
Thanks to our assumptions, every vertex of $V(G)$ is initially associated with a label of the halfspaces of $E_i$. The vertices $u' \in H_i'$ are labeled with $\Lambda_{H_i'}(u')$ while the vertices $u'' \in H_i''$ are labeled with $\Lambda_{H_i''}(u'')$.

For the construction of $\Lambda_G$, we begin with the computation of the gate of each vertex thanks to the BFS of Lemma~\ref{le:compute_gates}. As a reminder, for any vertex $u' \in H_i'$, we determine its $H_i''$-gate $g(u') \in H_i''$, and conversely for each vertex $u'' \in H_i''$ we determine its $H_i'$-gate $g(u'')$. This operation is achieved in time $O(n\log n)$ by launching a BFS with a starting queue made up of $H_i'$, and a second one with a starting queue $H_i''$ (as it was done with the eccentricities, see the proof of Theorem~\ref{th:balanced_recursion}). At the end of the execution, for each $u \in V(G)$, we know all distances $d(u,g(u))$.

We are ready to describe the content of the DO $\Lambda_G$. For each $u \in V(G)$, the sequence of bits $\Lambda_G(u)$ contains successively:
\begin{enumerate}
    \item $\lacla(u)$: the index $i$ of the $\Theta$-class $E_i$, encoded by $\log_2 n$ bits, since $q \le n$,
    \item $\laside(u)$: the side of $u$ regarding $E_i$ encoded with one bit: either $u \in H_i'$ or $u \in H_i''$,
    \item $\lagate(u)$: the identity of its gate $g(u)$ through the opposite halfspace encoded with $\log_2 n$ bits,
    \item $\ladist(u)$: the distance $d(u,g(u))$ encoded with $\log_2 n$ bits,
    \item $\larec(u)$: the label of $u$ in its halfspace: if $u\in H_i'$, then we add $\Lambda_{H_i'}(u)$, otherwise $\Lambda_{H_i''}(u)$.
\end{enumerate}
In summary, label $\Lambda_G(u)$ is the concatenation of $3\log_2 n + 1$ preliminary information in addition with the label of $u$ in its $E_i$-halfspace: $\Lambda_G(u) = \lacla(u) \cdot \laside(u) \cdot \lagate(u) \cdot \ladist(u) \cdot \larec(u)$. Hence, its size is at most $O(\log_2 n) + \max \{\vert \Lambda_{H_i'} \vert, \vert \Lambda_{H_i''} \vert\}$. Moreover, the time needed for the construction of $\Lambda_G$ includes the BFS evoked in Lemma~\ref{le:compute_gates} but also the writing of sequence $\Lambda_G(u)$. The latter takes at most $O(\log_2 n + \max \{\vert \Lambda_{H_i'} \vert, \vert \Lambda_{H_i''} \vert\})$ and is executed for each vertex of $G$. 

Let us explain how any distance $d(u,v)$ can be retrieved thanks to the DO $\Lambda_G$. From the $\log_2 n$ first bits of $\lacla(u)$ (or $\lacla(v)$), we obtain the $\Theta$-class which was used as a separator. Then, we look at the next bit $\laside$ for each vertex, it allows us to know into which halfspace $u$ and $v$ are. If they belong to the same halfspace , say $H_i'$ w.l.o.g., then we obtain $d(u,v)$ thanks to the DO $\Lambda_{H_i'}$, given by both $\larec(u)$ and $\larec(v)$. Else, if $u \in H_i'$ and $v \in H_i''$ w.l.o.g., we first retrieve the gate $g(u)$ of $u$ in $H_i''$ thanks to the part $\lagate(u)$, together with the distance $d(u,g(u))$ thanks to $\ladist(u)$. Second, we determine the distance $d(g(u),v)$ with a query on the DO $\Lambda_{H_i''}$, given by both $\larec(g(u))$ and $\larec(v)$. Finally, we obtain $d(u,v) = d(u,g(u)) + d(g(u),v)$.

The query time consists in the analysis of the first $3\log_2 n + 1$ bits to retrieve the different information: $E_i$, halfspace of each vertex, gate of $u$ and $d(u,g(u))$. Then, we necessarily launch a query on either $\Lambda_{H_i'}$ or $\Lambda_{H_i''}$. All in all, the query time is upper-bounded by $O(\log_2 n) + \max \{\tau_1(\vert H_i' \vert), \tau_2(\vert H_i'' \vert) \}$.
\end{proof}

We formulate an observation similar to the one we gave for the eccentricities problem after Theorem~\ref{th:balanced_recursion}. We could apply recursively Theorem~\ref{th:balanced_do} until we find base cases, {\em i.e.} median graphs of at most $2$ vertices. Unfortunately, if we select the $\Theta$-class $E_i$ arbitrarily, such an approach can lead to labels of linear size. However, in an utopian situation, we could select a balanced $\Theta$-class $E_i$ at each recursive step. In this way, the depth of the recursive tree would be logarithmic. The issue is that many median graphs do not admit any balanced $\Theta$-class. For this reason, we provide in the next subsection a labeling procedure for graphs $G \in \ut$.

\subsection{Oracle construction for median graphs without balanced $\Theta$-classes}

In order to handle the case of median graphs without any balanced $\Theta$-class, we label each vertex with not only its distance to the median vertex $v_0$, but also its gates for the ``neighboring'' fibers around it.
In this way, one can retrieve any distance $d(u,v)$ by looking at the successive distances from gate to gate a logarithmic number of times. The query time is thus poly-logarithmic, such as the size of the labeling.

Consider a median graph $G \in \ut$. According to Lemma~\ref{le:unique_median}, there is a unique median vertex $v_0$ for $G$ belonging to all majority halfspaces. Each vertex $v \neq v_0$ admits a ladder set $L_{v_0,v}$, hence the vertices of $G$ can be partitioned in function of their ladder $L_{v_0,v}$. We denote by $V_L \subseteq V(G) \setminus \{v_0\}$ the set of vertices $v$ with ladder set $L_{v_0,v} = L$. Set $V_L$ can also be seen, from Definition~\ref{def:ladder}, as the intersection of all minority halfspaces of $\Theta$-classes in $L$: $$V_L = \bigcap_{E_i \in L} H_i'.$$
For any POF $L$ containing only $\Theta$-classes adjacent to $v_0$, $V_L$ is not only nonempty but also gated, as shown in the next lemma.

\begin{lemma}
    For any POF $L$ adjacent to $v_0$, $V_L$ is nonempty and gated.
    \label{le:ladder_gated}
\end{lemma}
\begin{proof}
 As all $\Theta$-classes of $L$ are adjacent to $v_0$, there is a hypercube $Q_L$ containing both $v_0$ and edges of $L$ adjacent to $v_0$ (Lemma~\ref{le:pof_adjacent}). We denote by $v_L$ the farthest-to-$v_0$ vertex of $Q_L$, said differently the opposite vertex of $v_0$ in $Q_L$. By definition, the ladder set of $v_L$ is $L$, therefore $V_L$ is nonempty.
 
 Let $\stz(v_0)$ be the subgraph of $G$ made up of the vertices which belong to a common induced hypercube with $v_0$. This notion was defined in~\cite{ChLaRa19} and called the \emph{star} of a vertex. The authors (Proposition 2, \cite{ChLaRa19}) proved that, for any vertex $v$, $\stz(v)$ is a gated subgraph of $G$. Consequently, the fibers of $\stz(v_0)$ are gated (Lemma~\ref{le:fibers_gated}).
 
 We claim that $V_L$ is exactly the fiber of vertex $v_L$: in brief, $V_L = F_{\scriptsize{\stz(v_0)}}[v_L]$. Let $u \in V_L$ and $z \in \stz(v_0)$: we show that $v_L \in I(u,z)$ necessarily. If $\sigma_{v_L,u} \cap \sigma_{v_L,z} = \emptyset$, then our claim holds, according to Lemma~\ref{le:signature}. By contradiction, assume that some $\Theta$-class $E_j$ belongs to both $\sigma_{v_L,u}$ and $\sigma_{v_L,z}$. As $E_j \in \sigma_{v_L,z}$, then by convexity of $\stz(v_0)$, $E_j$ is adjacent to $v_0$. However, at the same time, we have $L = L_{v_0,u} = \sigma_{v_0,v_L} \subseteq \sigma_{v_0,u}$. Hence, $v_L \in I(v_0,u)$ and $\sigma_{v_L,u} \subseteq \sigma_{v_0,u}$ does not contain any $\Theta$-class adjacent to $v_0$ since all of them are present, by definition, in $\sigma_{v_0,v_L}$. So, $E_j$ is not adjacent to $v_0$, which contradicts our first observation. As a consequence, $v_L$ is the $\stz(v_0)$-gate of any $u \in V_L$.
As fibers are gated (Lemma~\ref{le:fibers_gated}), then we conclude that set $V_L$ is gated.
\end{proof}

The idea beyond our DO follows. We do not give any label to the median vertex $v_0$. We exploit the partition $\{ V_L : L\neq \emptyset,~\mbox{POF adjacent to}~ v_0\}$ of $V(G) \setminus \{v_0\}$. We compute recursively the DO of all median subgraphs $V_L$. Observe that, for any $\Theta$-class $E_i \in L$, then $V_L \subseteq H_i'$ and hence each set contains at most $\frac{n}{3}$ vertices, since $G \in \ut$. Then, we complete each label $\Lambda_{V_L}(u)$ with the identity of the ladder set $L = L_{v_0,u}$, the distance $d(v_0,u)$, and the pair gate-distance from $u$ to its gate in any neighboring fiber $V_{L'}$ of $V_L$. The formal (recursive) definition of the DO $\Lambda_G$ for $G\in \ut$ is given below.

\begin{definition}[DO $\Lambda_G$ for $G\in \ut$]
    Let $G \in \ut$ with $n\ge 3$. The label $\Lambda_G(v_0)$ is empty. For $u \in V(G) \setminus \{v_0\}$, sequence $\Lambda_G(u)$ contains:
    \begin{enumerate}
        \item $\laarg{L}(u)$: the ladder set $L_{v_0,u}$, encoded with $(\log_2 n)^2$ bits,
        \item $\ladiste(u)$: the distance $d(v_0,u)$, encoded with $\log_2 n$ bits,
        \item $\lafibe(u)$: each triplet $\left(L \setminus L',g_{L'}(u),d(u,g_{L'}(u))\right)$, for all subsets $L' \subsetneq L = L_{v_0,u}$ with $\vert L' \vert = \vert L \vert - 1$, where $g_{L'}(u)$ is the gate of $u$ for the gated fiber $V_{L'}$. Each triplet is encoded with $3\log_2 n$ bits, so the sequence of all triplets has size $O((\log_2 n)^2)$,
        \item $\larece(u)$: the label $\Lambda_{V_L}(u)$ of the DO $\Lambda_{V_L}$ computed recursively on the median subgraph $G[V_L]$.
    \end{enumerate}
    We have $\Lambda_G(u) = \laarg{L}(u) \cdot \ladiste(u) \cdot \lafibe(u) \cdot \larece(u)$.
    \label{def:do_unbalanced}
\end{definition}

\begin{figure}[h]
\centering
\scalebox{0.7}{\begin{tikzpicture}

\coordinate (s01) at (1.0,4.6) {};
\coordinate (s02) at (3.2,3.2) {};
\coordinate (s03) at (8.3,3.2) {};
\coordinate (s04) at (8.3,6.3) {};
\coordinate (s05) at (0.2,6.3) {};
\coordinate (s11) at (4.0,3.0) {};
\coordinate (s12) at (4.0,-0.8) {};
\coordinate (s13) at (9.3,-0.8) {};
\coordinate (s14) at (9.3,1.2) {};
\coordinate (s15) at (8.3,1.2) {};
\coordinate (s16) at (8.3,3.0) {};
\coordinate (s17) at (4.0,6.3) {};
\coordinate (s21) at (3.8,-0.8) {};
\coordinate (s22) at (-2.0,-0.8) {};
\coordinate (s23) at (-2.0,6.3) {};
\coordinate (s24) at (-0.1,6.3) {};
\coordinate (s25) at (0.8,4.4) {};
\coordinate (s26) at (3.2,3.0) {};
\coordinate (s27) at (3.8,3.0) {};





\node[scale=2, rotate=-60] at (-0.2,2) {{\bf \vdots}};
\node[scale=2, rotate=90] at (5.7,2.3) {{\bf \vdots}};
\node[scale=2, rotate=-40] at (6.0,5.2) {{\bf \vdots}};
\node[scale=2, rotate=60] at (-1.8,4.6) {{\bf \vdots}};

\draw[rounded corners, color = cyan, fill = white!90!cyan, line width = 0.8pt] (1.3,5.3) -- (-0.5,7.4) -- (2.5,7.4) -- (1.3,5.3);
\node[scale=1.3] at (1.0,7.8) {$V_{\{E_1,E_3\}}$};
\draw[rounded corners, color = blue, fill = white!90!blue, line width = 0.8pt] (3,4) -- (3.2,7.0) -- (5.2,6.6) -- (3,4);
\node[scale=1.3] at (4.4,7.3) {$V_{\{E_1\}}$};
\draw[rounded corners, color = green, fill = white!90!green, line width = 0.8pt] (1.3,3.8) -- (0.3,5.6) -- (-2.0,5.7) -- (1.3,3.8);
\node[scale=1.3] at (-1.0,6.1) {$V_{\{E_3\}}$};


\node[draw, circle, minimum height=0.2cm, minimum width=0.2cm, fill=black] (P11) at (1,1) {};
\node[draw, circle, minimum height=0.2cm, minimum width=0.2cm, fill=black] (P12) at (1,2.5) {};

\node[draw, circle, minimum height=0.2cm, minimum width=0.2cm, fill=black] (P21) at (3,1) {};
\node[draw, circle, minimum height=0.2cm, minimum width=0.2cm, fill=black] (P22) at (3,2.5) {};
\node[draw, circle, minimum height=0.2cm, minimum width=0.2cm, fill=black] (P23) at (3,4) {};

\node[draw, circle, minimum height=0.2cm, minimum width=0.2cm, fill=black] (P31) at (5,1) {};
\node[draw, circle, minimum height=0.2cm, minimum width=0.2cm, fill=black] (P32) at (5,2.5) {};
\node[draw, circle, minimum height=0.2cm, minimum width=0.2cm, fill=black] (P33) at (5,4) {};

\node[draw, circle, minimum height=0.2cm, minimum width=0.2cm, fill=black] (P41) at (1.3,3.8) {};
\node[draw, circle, minimum height=0.2cm, minimum width=0.2cm, fill=black] (P42) at (1.3,5.3) {};
\node[draw, circle, minimum height=0.2cm, minimum width=0.2cm, fill=black] (P43) at (-0.7,3.8) {};


\node[draw, circle, minimum height=0.2cm, minimum width=0.2cm, fill=black] (P52) at (1.1,6.8) {};

\node[draw, circle, minimum height=0.2cm, minimum width=0.2cm, fill=black] (P61) at (3.6,6.0) {};
\node[draw, circle, minimum height=0.2cm, minimum width=0.2cm, fill=black] (P62) at (0.1,5.2) {};

\coordinate (i1) at (2.2,7.2) {};
\coordinate (i2) at (1.2,5.7) {};


\node[scale=1.5] at (0.8,7.1) {$u$};
\node[scale=1.5] at (2.7,2.2) {$v_0$};
\node at (4.2,6.4) {$g_{\{E_1\}}(u)$};
\node at (-0.8,5.4) {$g_{\{E_3\}}(u)$};
\node[scale=1.5,color = blue] at (5.4,3.2) {$E_1$};
\node[scale=1.5,color = red] at (4.0,1.4) {$E_2$};
\node[scale=1.5,color = green] at (0,2.8) {$E_3$};

\draw[line width = 1.8pt, dashed, out = -20, in = 160] (P52) to (i1);
\draw[line width = 1.8pt, dashed] (i1) -- (P61);
\draw[line width = 1.8pt, out = -135, in = 90, dashed] (P52) to (i2);
\draw[line width = 2pt, dashed] (i2) -- (P62);


\draw[line width = 1.4pt] (P11) -- (P12);
\draw[line width = 1.4pt] (P11) -- (P21);
\draw[line width = 1.4pt] (P12) -- (P22);
\draw[line width = 1.4pt] (P21) -- (P22);

\draw[line width = 1.8pt, color = red] (P21) -- (P31);
\draw[line width = 1.8pt, color = red] (P22) -- (P32);
\draw[line width = 1.4pt] (P31) -- (P32);

\draw[line width = 1.8pt, color = blue] (P22) -- (P23);
\draw[line width = 1.8pt, color = red] (P23) -- (P33);
\draw[line width = 1.8pt, color = blue] (P32) -- (P33);

\draw[line width = 1.8pt, color = green] (P22) -- (P41);
\draw[line width = 1.8pt, color = green] (P12) -- (P43);
\draw[line width = 1.8pt, color = green] (P23) -- (P42);
\draw[line width = 1.4pt] (P41) -- (P43);
\draw[line width = 1.8pt, color = blue] (P41) -- (P42);

\node[scale=1.5] at (13.5,4) {$\begin{array}{rl}
\Lambda_G(u) = & \{E_1,E_3\} ~~{\color{red} : (\log_2 n)^2 ~\mbox{bits}}\\ 
& d(v_0,u) ~~~{\color{red} : \log_2 n ~\mbox{bits}}\\
& \left(\{E_1\},g_{\{E_1\}}(u), d(u,g_{\{E_1\}}(u)\right)\\
& \left(\{E_3\},g_{\{E_3\}}(u), d(u,g_{\{E_3\}}(u)\right)\\
& \Lambda_{V_{\{E_1,E_3\}}}(u)
\end{array}$};

\node[scale = 3, color = red] at (10.6,3.7) {$\{$};
\node[scale = 1.5, text width = 10mm, color = red] at (8.0,3.7) {$\lafib(u):$ ~~~~$O((\log_2 n)^2)$ bits};

\end{tikzpicture}}
\caption{The label $\Lambda_G(u)$ of some vertex $u \in V(G)$, $G \in \ut$}
\label{fig:do_unbalanced}
\end{figure}
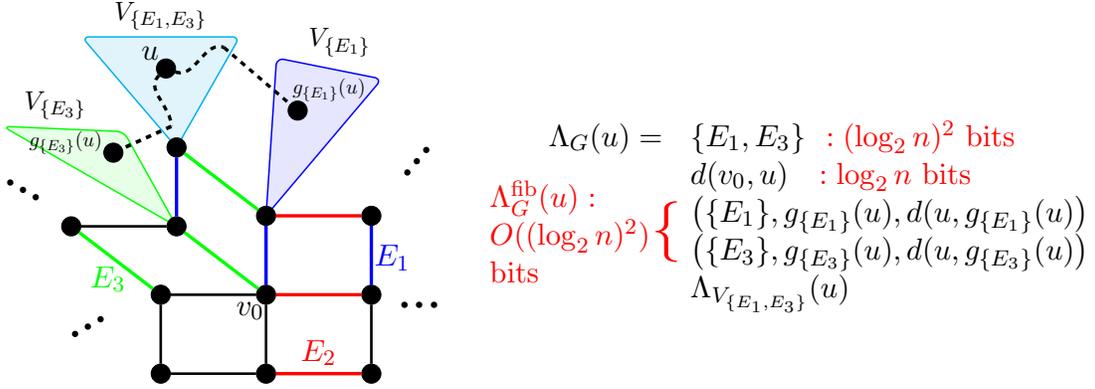

An example of label $\Lambda_G(u)$ is given in Figure~\ref{fig:do_unbalanced}. The definition of our labeling $\Lambda_G$ is now complete. Remember that, if the graph $G$ has at most $2$ vertices, $\vert \Lambda_G \vert = 0$: there is no need to label the vertices since they are necessarily at distance $1$. Else, we distinguish two cases. If $G$ contains a $3$-balanced $\Theta$-class $E_i$, then we compute recursively the DO of each halfspace and label a logarithmic number of bits for each vertex in addition with the DO of its halfspace, as described in Theorem~\ref{th:balanced_do}. If $G$ contains no $3$-balanced $\Theta$-class, {\em i.e.} $G\in \ut$, then we proceed as in Definition~\ref{def:do_unbalanced}: we compute recursively the DO of each fiber $V_L$ and label each vertex with a poly-logarithmic number of bits in addition with the DO of its fiber $V_L$. We thus conclude on the size of the labeling $\Lambda_G$.

\begin{lemma}
  For any median graph $G$, the size of $\Lambda_G$ is poly-logarithmic: $\vert \Lambda_G \vert = O((\log_2 n)^3)$
  \label{le:do_size}
\end{lemma}
\begin{proof}
Let $\alpha(n)$ be the maximum size of a DO $\Lambda_G$ on median graphs $G$ satisfying $\vert V(G) \vert = n$. For $n \ge 3$, we show that $\alpha(n) \le 4(\log_2 n)^2 + \log_2 n + \alpha(\lfloor \frac{2n}{3} \rfloor)$. 

If $G$ admits a $3$-balanced $\Theta$-class, as shown in Theorem~\ref{th:balanced_do}, the size of the labeling $\Lambda_G$ is at most $3\log_2 n + 1 + \alpha(\lfloor \frac{2n}{3} \rfloor)$ since each halfspace of a $3$-balanced $\Theta$-class contains at most $\lfloor \frac{2n}{3} \rfloor$ vertices. 

However, if $G \in \ut$, let us justify that, in this case, the size of the labeling $\Lambda_G$ is upper-bounded by $4(\log_2 n)^2 + \log_2 n + \alpha(\lfloor \frac{n}{3} \rfloor)$. As already stated in Definition~\ref{def:do_unbalanced}, the size of the first two parts of the labeling, {\em i.e.} $\laarg{L}$ and $\ladist$, is $(\log_2 n)^2 + \log_2 n$. The third part $\lafib$ is trickier to analyze. Let $u \in V(G)$ and $L = L_{v_0,u}$. The label $\lafib(u)$ contains as many triplets as the number of subsets $L'$ of $L$ with exactly one element left. As $\vert L \vert \le d \le \log_2 n$, there are at most $\log_2 n$ such subsets. The triplet is made up of the singleton $L\setminus L'$,  the gate $g_{L'}(u)$ and the distance $d(u,g_{L'}(u))$. Each component of the triplet can be encoded with $\log_2 n$ bits. Therefore, the size of $\lafib(u)$ is at most $3(\log_2 n)^2$. Finally, $\Lambda_G(u)$ also contains the part $\larec(u)$, which is the label of $u$ for the DO on graph $G[V_L]$. As $\vert V_L \vert < \frac{n}{3}$, we have $\vert \Lambda_G(u) \vert \le 4(\log_2 n)^2 + \log_2 n + \alpha(\lfloor \frac{n}{3} \rfloor)$.

Combining the inequalities obtained for the two possible cases, {\em i.e.} $G\in \ut$ or $G \notin \ut$, we have: $\alpha(n) \le 4(\log_2 n)^2 + \log_2 n + \alpha(\lfloor \frac{2n}{3} \rfloor)$. Furthermore, as a base case, $\alpha(1) = \alpha(2) = 0$. The depth of this recursive sequence is logarithmic since we divide at each recursive step the number of vertices of the graph by a constant factor smaller than $1$. The extra size added into the labeling at each recursive step is $O((\log_2 n)^2)$. We conclude: $\alpha(n) = O((\log_2 n)^3)$.
\end{proof}

The size of labeling $\Lambda_G$ for any median graph $G$ is guaranteed to be poly-logarithmic. We show now that obtaining labels $\Lambda_G(u)$ for each vertex $u \in V(G)$ can be achieved in quasilinear time.

\begin{lemma}
 There is a combinatorial algorithm which, given any median graph $G$, outputs the labeling $\Lambda_G$ in quasilinear time $O(n(\log_2 n)^4)$.
 \label{le:do_construct}
\end{lemma}
\begin{proof}
Let $T(n)$ be the maximum construction time of a labeling $\Lambda_G$ for median graphs $G$ satisfying $\vert V(G) \vert = n$. We almost entierely handled the case where $G \notin \ut$: in Theorem~\ref{th:balanced_do}, we proved that, given a balanced $\Theta$-class $E_i$ and the labeling $\Lambda_{H_i'}$ and $\Lambda_{H_i''}$ of its halfspaces, we build the labeling $\Lambda_G$ in time $O(n(\log_2 n + \max \{\vert \Lambda_{H_i'} \vert, \vert \Lambda_{H_i''} \vert\}))$. Now, we naturally add the construction time for labelings $\Lambda_{H_i'}$ and $\Lambda_{H_i''}$ which were supposed to be known in Theorem~\ref{th:balanced_do}. Moreover, thanks to Lemma~\ref{le:do_size}, values $\vert \Lambda_{H_i'} \vert$ and $\vert \Lambda_{H_i''} \vert$ can be upper-bounded by $O((\log_2 n)^3)$. In summary, we obtain: $T(n) \le O(n(\log_2 n)^3) + T(\lambda n) + T(\mu n)$, where $\lambda = \frac{\vert H_i' \vert}{n}$ and $\mu = \frac{\vert H_i'' \vert}{n}$, hence $\lambda + \mu = 1$ and $\max\{\lambda,\mu\} \le \frac{2}{3}$.

We focus now on the case $G \in \ut$. Let $u \in V(G) \setminus \{v_0\}$. First, a BFS starting from $v_0$ gives us the ladder $L_{v_0,u}$ (Lemma~\ref{le:ladder_linear}) together with the distance $d(v_0,u)$. Second, we present a procedure to obtain all labels $\lafib(u)$. We begin with looking at all edges of the graph: if they connect a vertex of a fiber $V_L$ with another vertex of a fiber $V_{L'}$, where $L' \subsetneq L$, $\vert L' \vert = \vert L \vert - 1$, we associate this edge to the pair $(L,L')$. We denote by $\partial E_{L,L'}$ this kind of edges and by $\partial V_{L,L'}$ the set of vertices of $V_{L'}$ which admit a neighbor in $V_L$.

For any POF $L$ adjacent to $v_0$ and any subset $L' \subsetneq L$, $\vert L' \vert = \vert L \vert - 1$, we launch a BFS on graph $G_{L,L'} = (V_L \cup \partial V_{L,L'}, E[V_L] \cup \partial E_{L,L'})$ in order to determine, for each $u \in V_L$, its $V_{L'}$-gate $g_{L'}(u)$ and also $d(u,g_{L'}(u))$ (Lemma~\ref{le:compute_gates}). The running time of this operation depends on the total number of edges of all graphs $G_{L,L'}$. Any edge of $V_L$ belongs to at most $d \le \log_2 n$ such graphs, as there are $\vert L \vert$ possible subsets $L'$ of $L$. The edges connecting two sets $V_L$ and $V_{L'}$ are thus taken into account in only one graph $G_{L,L'}$. All in all, as each edge is considered in at most $\log_2 n$ such graphs $G_{L,L'}$, the total running time of this procedure is $O(m\log_2 n) = O(n(\log_2 n)^2)$. With the writing of each label, the total time spent to determine blocks $\laarg{L}$, $\ladist$ and $\lafib$ is $O(n(\log_2 n)^3)$. Finally, we need to compute recursively the DO of each fiber $V_L$. In brief, $T(n) \le O(n(\log_2 n)^3) + \sum_{L} T(\vert V_L \vert)$. As each $\vert V_L \vert < \frac{n}{3}$, we can write: $T(n) \le O(n(\log_2 n)^3) + \sum_k T(\alpha_k n)$, where $\sum_k \alpha_k = 1$ and $\max \{\alpha_k\} < \frac{1}{3}$.

Both inequations for $T(n)$ fall into the domain of the Master theorem described in~\cite{CoLeRiSt09}. At each recursive step, an extra time $O(n(\log_2 n)^3)$ is needed and, at the same time, the size of the instances on which recursive calls are launched is decreased by a constant factor. Moreover, the total number of vertices considered at each depth of the recursive tree stays equal to $n$. All in all, we have: $T(n) = O(n(\log_2 n)^4)$.
\end{proof}

The last part of our proof is certainly the most important of course: even if the labeling $\Lambda_G$ has poly-logarithmic size and can be computed in quasilinear time, we do not know yet whether it allows us to retrieve any distance in poly-logarithmic time. The algorithm we present below use a new technique for the case $G \in \ut$: to retrieve distance $d(u,v)$, we compare the ladder sets of both vertices, {\em i.e.} $L_{v_0,u}$ and $L_{v_0,v}$. The idea consists in going from $u$ to $v$ via "neighboring" ladder sets (differing from only one $\Theta$-class), using the information of labeling $\lafib$.

\begin{theorem}
 Let $G$ be a median graph and $u,v \in V(G)$. One can retrieve any distance $d(u,v)$ thanks to $\Lambda_G$ with poly-logarithmic query time $O(\log^4(n))$.
 \label{th:do_query}
\end{theorem}
\begin{proof}
For graphs with at most two vertices, one can retrieve distances in constant time. We proceed by induction on the size of $\vert V(G) \vert$ to prove that labeling $\Lambda_G$ is a DO. For median graphs $G \notin \ut$: according to Theorem~\ref{th:balanced_do}, if the DOs $\Lambda_{H_i'}$ and $\Lambda_{H_i''}$ are known, then the labeling $\Lambda_G$ we built is a DO. By induction hypothesis, our labeling is a DO for graphs with a smaller number of vertices, in particular $G[H_i']$ and $G[H_i'']$. So, $\Lambda_G$ is a DO for $G \notin \ut$ with query time $\tau(n) = O(\log_2 n) + \max \{\tau(\vert H_i' \vert), \tau(\vert H_i'' \vert)\} \le O(\log_2 n) + \tau(\lfloor \frac{2n}{3} \rfloor)$.

Let us focus on the more complicated case $G \in \ut$: our objective is to retrieve a distance $d(u,v)$, for some pair $u,v \in V(G)$. A first simple case occurs when one of these vertices is $v_0$, w.l.o.g. $u = v_0$. The label $\ladist(v)$, which is a part of $\Lambda_G(v)$, gives directly distance $d(v_0,v)$. Hence, we retrieve this distance in constant time. 

Assume that both $u$ and $v$ are different from $v_0$. Both admit a ladder set with $v_0$, which can be deduced from $\laarg{L}(u)$ and $\laarg{L}(v)$. Say, w.l.o.g., that $\vert L_{v_0,u}\vert \le \vert L_{v_0,v}\vert$. Another simple case is when $L_{v_0,u} = L_{v_0,v}$. The distance $d(u,v)$ can be computed thanks to a query on the labels $\Lambda_{V_L}(u)$ and $\Lambda_{V_L}(v)$, where $L = L_{v_0,u}$. Indeed, $V_L$ is gated, so $G[V_L]$ is a median subgraph of $G$. These two labels are given by the part $\larec$ of $\Lambda_G$ over each vertex. The query time in this situation is thus $\tau(\lfloor \frac{n}{3} \rfloor)$ by induction hypothesis, as $\vert V_L \vert < \frac{n}{3}$.

Nevertheless, in general, $L_{v_0,u} \neq L_{v_0,v}$. A third case easy to handle is when $L_{v_0,u} \cap L_{v_0,v} = \emptyset$. In this scenario, the signatures $\sigma_{v_0,u}$ and $\sigma_{v_0,v}$ have no $\Theta$-class in common: if they had one, say $E_i$, by convexity of its boundary $\partial H_i''$, then $v_0 \in \partial H_i''$ and $E_i$ must belong to both ladder sets from Definition~\ref{def:ladder}. So, $v_0 \in I(u,v)$ and the distance $d(u,v) = d(u,v_0) + d(v_0,v)$ can be retrieved from labels $\ladist(u)$ and $\ladist(v)$.

The remaining (and hard) case is the following one: $L_{v_0,u} \neq L_{v_0,v}$ and $L_{v_0,u} \cap L_{v_0,v} \neq \emptyset$. First, we compute the \textit{ladder sequence} $\mathcal{S}_{u,v}$ between these two sets $L_{v_0,u}$ and $L_{v_0,v}$: this a finite sequence of POFs, starting from $L_{v_0,u}$ and finishing at $L_{v_0,v}$. The size gap between two consecutive POFs of this sequence is exactly $1$. The first part of $\mathcal{S}_{u,v}$ starts with $L_{v_0,u}$ and goes towards $L_{v_0,u} \cap L_{v_0,v}$. At each step, a $\Theta$-class of $L_{v_0,u} \setminus L_{v_0,v}$ is withdrawn from the current POF. The second part of $\mathcal{S}_{u,v}$ starts with $L_{v_0,u} \cap L_{v_0,v}$ and goes towards $L_{v_0,v}$ by adding to the current POF a $\Theta$-class of $L_{v_0,v} \setminus L_{v_0,u}$ at each step. The elements of sequence $\mathcal{S}_{u,v}$ are denoted by: $$\mathcal{S}_{u,v} = \left(L^{(-r)},L^{(-r+1)},\ldots,L^{(-1)},L^{(0)},L^{(1)},\ldots,L^{(t-1)},L^{(t)}\right),$$
where $r$ and $t$ are nonnegative integers, $L^{(-r)} = L_{v_0,u}$, $L^{(t)} = L_{v_0,v}$, and $L^{(0)} = L_{v_0,u} \cap L_{v_0,v}$. For any integer $0\le i \le r-1$, $L^{(-i+1)}$ is obtained from $L^{(-i)}$ by withdrawing one $\Theta$-class belonging to $L^{(-i)} \setminus L_{v_0,v}$. Similarly, for any $0\le j \le t-1$, $L^{(j)}$ is obtained from $L^{(j+1)}$ by withdrawing one $\Theta$-class of $L^{(j+1)} \setminus L_{v_0,u}$. As an example, if $L_{v_0,u} = \{E_1,E_2,E_3\}$ and $L_{v_0,v} = \{E_3,E_4,E_5\}$, then $\mathcal{S}_{u,v}$ might be: $(\{E_1,E_2,E_3\},\{E_2,E_3\},\{E_3\},\{E_3,E_4\},\{E_3,E_4,E_5\})$. As each ladder set has size at most $d$, the length of sequence $\mathcal{S}_{u,v}$ is upper-bounded by $2d$.

For any integer $0\le i\le r-1$, we denote by $E^{(-i)}$ the $\Theta$-class of the singleton $L^{(-(i+1))} \setminus L^{(-i)}$. Similarly, $E_{(j)}$ denotes the $\Theta$-class of $L^{(j+1)} \setminus L^{(j)}$, for $0\le j\le t-1$. We show that there is a shortest $(u,v)$-path passing through all sets $V_L$ for $L \in \mathcal{S}_{u,v}$. We proceed in three steps, each one characterized by a claim.

\textit{Claim 1}: the median of triplet $u,v,v_0$, vertex $m = m(u,v,v_0)$ (see Definition~\ref{def:median}), belongs to $V_{L^{(0)}}$. As $m \in I(u,v)$, then it belongs, by convexity, to the minority halfspaces of all $\Theta$-classes in $L_{v_0,u} \cap L_{v_0,v}$. But, as $m \in I(v,v_0)$ and both $v$ and $v_0$ are into the majority halfspace of all $\Theta$-classes of $L_{v_0,u} \setminus L_{v_0,v}$, then $m$ is also in the majority halfspace of these $\Theta$-classes. Observe that the same argument holds for $L_{v_0,v} \setminus L_{v_0,u}$, because $m \in I(u,v_0)$. In summary, $m$ belongs to the majority halfspace of all $\Theta$-classes in the symmetric difference between $L_{v_0,u}$ and $L_{v_0,v}$, but to the minority halfspace of all $\Theta$-classes in $L_{v_0,u} \cap L_{v_0,v}$. Therefore, $L_{v_0,m} = L_{v_0,u} \cap L_{v_0,v}$.

\textit{Claim 2}: there exists a shortest $(u,m)$-path passing through all fibers $V_{L^{(-i)}}$, $0\le i\le r$. We show briefly that a way to go from $u$ to $m$ is to traverse fiber $V_{L^{(-r+1)}}$. We will see that, applying the same argument iteratively, leads to the same conclusion than Claim 2. Vertex $u$ belongs to the minority halfspace of $E^{(-r+1)}$ while $m$ belongs to its majority halfspace. So, $u$ admits a gate $g_{-r+1}(u)$ for the majority halfspace of $E^{(-r+1)}$. Vertex $g_{-r+1}(u)$ belongs to $V_{L^{(-r+1)}}$. Then, we can pursue with the same argument for the gate $g_{-r+1}(u)$: it admits a gate $g_{-r+2}(u)$ in the majority halfspace of $E^{(-r+2)}$, and so on. Finally, we produce a sequence of gates, all belonging by definition to $I(u,m)$. The last gate of this sequence, $g_0(u)$, belongs to $V_{L^{(0)}}$, such as $m$, hence the last part of the shortest $(u,m)$-path stays in $V_{L^{(0)}}$, by convexity of fibers.

\textit{Claim 3}: there exists a shortest $(v,m)$-path passing through all fibers $V_{L^{(j)}}$, $0\le j\le t$. The argument is the same as the one for Claim 2: we start from $v$ and determine iteratively the successive gates on the majority halfspaces of respectively $\Theta$-classes $E^{(t-1)},E^{(t-2)},\ldots$

As a conclusion of the three claims, since $m \in I(u,v)$ belongs to $V_{L^{(0)}}$, we ensure that there exists a shortest $(u,v)$-path, consisting first in a section from $u$ to $m$ passing through all fibers indexed negatively (Claim 2), second in a section from $m$ to $v$ passing through all fibers indexed positively (Claim 3). In brief, there is a shortest $(u,v)$-path passing successively through the fibers $V_L$, where $L$ follows the sequence $\mathcal{S}_{u,v}$. An example is provided in Figure~\ref{fig:retrieve_dist}, with a sequence $\mathcal{S}_{u,v} = (\{E_1,E_3\},\{E_1\},\{E_1,E_2\})$. The orange dashed path represents a shortest $(u,v)$-path traversing all successive fiber gates.

\begin{figure}[h]
\centering
\scalebox{0.7}{\begin{tikzpicture}

\coordinate (s01) at (1.0,4.6) {};
\coordinate (s02) at (3.2,3.2) {};
\coordinate (s03) at (8.3,3.2) {};
\coordinate (s04) at (8.3,6.3) {};
\coordinate (s05) at (0.2,6.3) {};
\coordinate (s11) at (4.0,3.0) {};
\coordinate (s12) at (4.0,-0.8) {};
\coordinate (s13) at (9.3,-0.8) {};
\coordinate (s14) at (9.3,1.2) {};
\coordinate (s15) at (8.3,1.2) {};
\coordinate (s16) at (8.3,3.0) {};
\coordinate (s17) at (4.0,6.3) {};
\coordinate (s21) at (3.8,-0.8) {};
\coordinate (s22) at (-2.0,-0.8) {};
\coordinate (s23) at (-2.0,6.3) {};
\coordinate (s24) at (-0.1,6.3) {};
\coordinate (s25) at (0.8,4.4) {};
\coordinate (s26) at (3.2,3.0) {};
\coordinate (s27) at (3.8,3.0) {};





\node[scale=2, rotate=-60] at (-0.2,2) {{\bf \vdots}};
\node[scale=2, rotate=90] at (5.7,2.3) {{\bf \vdots}};
\node[scale=2, rotate=60] at (-1.8,4.6) {{\bf \vdots}};

\draw[rounded corners, color = cyan, fill = white!90!cyan, line width = 0.8pt] (1.3,5.3) -- (-0.5,7.4) -- (2.5,7.4) -- (1.3,5.3);
\node[scale=1.3] at (1.0,7.8) {$V_{\{E_1,E_3\}}$};
\draw[rounded corners, color = blue, fill = white!90!blue, line width = 0.8pt] (3,4) -- (3.2,7.0) -- (5.2,6.6) -- (3,4);
\node[scale=1.3] at (4.4,7.3) {$V_{\{E_1\}}$};
\draw[rounded corners, color = green, fill = white!90!green, line width = 0.8pt] (1.3,3.8) -- (0.3,5.6) -- (-2.0,5.7) -- (1.3,3.8);
\node[scale=1.3] at (-1.0,6.1) {$V_{\{E_3\}}$};
\draw[rounded corners, color = purple, fill = white!90!purple, line width = 0.8pt] (5,4) -- (9.5,4.5) -- (5.9,7.6) -- (5,4);
\node[scale=1.3] at (10.4,4.6) {$V_{\{E_1,E_2\}}$};


\node[draw, circle, minimum height=0.2cm, minimum width=0.2cm, fill=black] (P11) at (1,1) {};
\node[draw, circle, minimum height=0.2cm, minimum width=0.2cm, fill=black] (P12) at (1,2.5) {};

\node[draw, circle, minimum height=0.2cm, minimum width=0.2cm, fill=black] (P21) at (3,1) {};
\node[draw, circle, minimum height=0.2cm, minimum width=0.2cm, fill=black] (P22) at (3,2.5) {};
\node[draw, circle, minimum height=0.2cm, minimum width=0.2cm, fill=black] (P23) at (3,4) {};

\node[draw, circle, minimum height=0.2cm, minimum width=0.2cm, fill=black] (P31) at (5,1) {};
\node[draw, circle, minimum height=0.2cm, minimum width=0.2cm, fill=black] (P32) at (5,2.5) {};
\node[draw, circle, minimum height=0.2cm, minimum width=0.2cm, fill=black] (P33) at (5,4) {};

\node[draw, circle, minimum height=0.2cm, minimum width=0.2cm, fill=black] (P41) at (1.3,3.8) {};
\node[draw, circle, minimum height=0.2cm, minimum width=0.2cm, fill=black] (P42) at (1.3,5.3) {};
\node[draw, circle, minimum height=0.2cm, minimum width=0.2cm, fill=black] (P43) at (-0.7,3.8) {};


\node[draw, circle, minimum height=0.2cm, minimum width=0.2cm, fill=black] (P51) at (8.4,4.8) {};
\node[draw, circle, minimum height=0.2cm, minimum width=0.2cm, fill=black] (P52) at (1.1,6.8) {};

\node[draw, circle, minimum height=0.2cm, minimum width=0.2cm, fill=black] (P61) at (3.6,6.0) {};
\node[draw, circle, minimum height=0.2cm, minimum width=0.2cm, fill=black] (P62) at (3.4,5.1) {};

\coordinate (i1) at (2.2,7.2) {};
\coordinate (i2) at (5.8,5.1) {};

\draw[line width = 1.8pt, dashed, color = orange, out = -20, in = 160] (P52) to (i1);
\draw[line width = 1.8pt, dashed, color = orange] (i1) -- (P61);
\draw[line width = 1.8pt, dashed, color = orange, out = -150, in = 0] (P51) to (i2);
\draw[line width = 1.8pt, dashed, color = orange] (i2) -- (P62);
\draw[line width = 1.8pt, dashed, color = orange] (P61) -- (P62);


\draw[line width = 1.4pt] (P11) -- (P12);
\draw[line width = 1.4pt] (P11) -- (P21);
\draw[line width = 1.4pt] (P12) -- (P22);
\draw[line width = 1.4pt] (P21) -- (P22);

\draw[line width = 1.8pt, color = red] (P21) -- (P31);
\draw[line width = 1.8pt, color = red] (P22) -- (P32);
\draw[line width = 1.4pt] (P31) -- (P32);

\draw[line width = 1.8pt, color = blue] (P22) -- (P23);
\draw[line width = 1.8pt, color = red] (P23) -- (P33);
\draw[line width = 1.8pt, color = blue] (P32) -- (P33);

\draw[line width = 1.8pt, color = green] (P22) -- (P41);
\draw[line width = 1.8pt, color = green] (P12) -- (P43);
\draw[line width = 1.8pt, color = green] (P23) -- (P42);
\draw[line width = 1.4pt] (P41) -- (P43);
\draw[line width = 1.8pt, color = blue] (P41) -- (P42);


\node[scale=1.5] at (0.8,7.1) {$u$};
\node[scale=1.5] at (8.2,5.2) {$v$};
\node[scale=1.5] at (2.7,2.2) {$v_0$};
\node at (4.2,6.4) {$g_{\{E_1\}}(u)$};
\node at (2.4,4.9) {$g_{\{E_1\}}(v)$};
\node[scale=1.5,color = blue] at (5.4,3.2) {$E_1$};
\node[scale=1.5,color = red] at (4.0,1.4) {$E_2$};
\node[scale=1.5,color = green] at (0,2.8) {$E_3$};

\end{tikzpicture}}
\caption{A concrete view of a median graph $G \in \ut$ and its fiber to show how we retrieve distance $d(u,v)$, with $L_{v_0,u} = \{E_1,E_3\}$ and $L_{v_0,v} = \{E_1,E_2\}$}
\label{fig:retrieve_dist}
\end{figure}

In order to retrieve distance $d(u,v)$, we use some labels of each fiber $V_L$, where $L \in \mathcal{S}_{u,v}$. From label $\lafib(u)$, we obtain the gate of $u$ for the set $V_{L^{(-r+1)}}$, which is necessarily $g_{-r+1}(u)$ since $V_{L^{(-r+1)}}$ is a subset of the majority halfspace of $E^{(-r+1)}$. Then, from $\lafib(g_{-r+1}(u))$, we obtain the second gate, which belongs to set $V_{L^{(-r+2)}}$, and so on. In summary, we follow the sequence $\mathcal{S}_{u,v}$ of logarithmic size and obtain from $\lafib$ successively the pair gate/distance in the next majority halfspace considered. By simply summing up all the distances between the different gates, we obtain $d(u,v)$. The distance that remains unknown even after the pick up of pairs gate/distance is the distance between the two gates obtain in fiber $V_{L^{(0)}}$ ({\em e.g.} vertices $g_{\{E_1\}}(u)$ and $g_{\{E_1\}}(v)$ in Figure~\ref{fig:retrieve_dist}). To compute the distance between these two final gates, a query on $\larec$ for these gates suffices. Hence, distance $d(u,v)$ is now determined.


As $u$ and $v$ belong to different fibers but with a nonempty intersection, the whole process consists in (i) retrieving the ladder set of both vertices thanks to $\laarg{L}$, (ii) computing the sequence $\mathcal{S}_{u,v}$ and then the successive gates thanks to $\lafib$, (iii) obtaining the distance between all successive gates thanks to $\lafib$, and (iv) computing the distance between the gates of $V_{L^{(0)}}$ thanks to $\larec$. The cost of (i) is negligible compared to the one of (ii) and (iii) together, which is $O((\log_2 n)^3)$, since the computation of one gate takes $O(\vert \lafib \vert) = O((\log_2 n)^2)$ and there are at most $2d$ of them. So, taking also (iv) into account, we have $\tau(n) = \tau(\lfloor \frac{n}{3}\rfloor) + O((\log_2 n)^3)$.


Considering all possible cases, an upper bound is $\tau(n) = \tau(\lfloor \frac{2n}{3}\rfloor) + O((\log_2 n)^3)$, consequently, by Master theorem, $\tau(n) = O((\log_2 n)^4)$.
\end{proof}

The labeling $\Lambda_G$ thus guarantees all the properties listed in the following statement, which is the second contribution of our paper.

\mainDO*

\bibliographystyle{plain}
\bibliography{median}

\begin{thebibliography}{10}

\bibitem{AbGrWi15}
A.~Abboud, F.~Grandoni, and V.~V. Williams.
\newblock Subcubic equivalences between graph centrality problems, {APSP} and
  diameter.
\newblock In {\em Proc. of {SODA}}, pages 1681--1697, 2015.

\bibitem{Ba84}
H.~Bandelt.
\newblock Retracts of hypercubes.
\newblock {\em Journal of Graph Theory}, 8(4):501--510, 1984.

\bibitem{BaBa84}
H.~Bandelt and J.~Barth{\'{e}}lemy.
\newblock Medians in median graphs.
\newblock {\em Discret. Appl. Math.}, 8(2):131--142, 1984.

\bibitem{BaCh08}
H.~Bandelt and V.~Chepoi.
\newblock Metric graph theory and geometry: a survey.
\newblock {\em Contemp. Math.}, 453:49--86, 2008.

\bibitem{BaChDrKo06}
H.~Bandelt, V.~Chepoi, A.~W.~M. Dress, and J.~H. Koolen.
\newblock Combinatorics of lopsided sets.
\newblock {\em Eur. J. Comb.}, 27(5):669--689, 2006.

\bibitem{BaChEp10}
H.~Bandelt, V.~Chepoi, and D.~Eppstein.
\newblock Combinatorics and geometry of finite and infinite squaregraphs.
\newblock {\em {SIAM} J. Discret. Math.}, 24(4):1399--1440, 2010.

\bibitem{BaFoRo99}
H.~Bandelt, P.~Forster, and A.~Röhl.
\newblock {Median-joining networks for inferring intraspecific phylogenies}.
\newblock {\em Molecular Biology and Evolution}, 16(1):37--48, 1999.

\bibitem{BaFoSyRi95}
H.~Bandelt, P.~Forster, B.~C. Sykes, and M.~B. Richards.
\newblock Mitochondrial portraits of human populations using median networks.
\newblock {\em Genetics}, 141(2):743--753, 1995.

\bibitem{BaMaRi00}
H.~Bandelt, V.~Macaulay, and M.~Richards.
\newblock Median networks: Speedy construction and greedy reduction, one
  simulation, and two case studies from human {mtDNA}.
\newblock {\em Molecular Phylogenetics and Evolution}, 16(1):8--28, 2000.

\bibitem{BaMu91}
H.~Bandelt and H.~M. Mulder.
\newblock Pseudo-median graphs: decomposition via amalgamation and cartesian
  multiplication.
\newblock {\em Discret. Math.}, 94(3):161--180, 1991.

\bibitem{BaQuSaMa02}
H.~Bandelt, L.~Quintana-Murci, A.~Salas, and V.~Macaulay.
\newblock The fingerprint of phantom mutations in mitochondrial {DNA} data.
\newblock {\em Am. J. Hum. Genet.}, 71:1150--1160, 2002.

\bibitem{BaCo93}
J.~Barth{\'{e}}lemy and J.~Constantin.
\newblock Median graphs, parallelism and posets.
\newblock {\em Discret. Math.}, 111(1-3):49--63, 1993.

\bibitem{BeBhShTa07}
B.~Ben{-}Moshe, B.~K. Bhattacharya, Q.~Shi, and A.~Tamir.
\newblock Efficient algorithms for center problems in cactus networks.
\newblock {\em Theor. Comput. Sci.}, 378(3):237--252, 2007.

\bibitem{BeChChVa20}
L.~B{\'{e}}n{\'{e}}teau, J.~Chalopin, V.~Chepoi, and Y.~Vax{\`{e}}s.
\newblock Medians in median graphs and their cube complexes in linear time.
\newblock In {\em Proc. of {ICALP}}, volume 168, pages 10:1--10:17, 2020.

\bibitem{BeDuHa22}
P.~Berg{\'{e}}, G.~Ducoffe, and M.~Habib.
\newblock Subquadratic-time algorithm for the diameter and all eccentricities
  on median graphs.
\newblock In {\em Procs. of {STACS}}, volume 219 of {\em LIPIcs}, pages
  9:1--9:21, 2022.

\bibitem{BeHa21}
P.~Berg{\'{e}} and M.~Habib.
\newblock Diameter, radius and all eccentricities in linear time for
  constant-dimension median graphs.
\newblock In {\em Proc. of LAGOS}, 2021.

\bibitem{Br07}
B.~Bre{\v{s}}ar.
\newblock Characterizing almost-median graphs.
\newblock {\em Eur. J. Comb.}, 28(3):916--920, 2007.

\bibitem{BrKlSk07}
B.~Bre{\v{s}}ar, S.~Klav{\v{z}}ar, and R.~Skrekovski.
\newblock On cube-free median graphs.
\newblock {\em Discret. Math.}, 307(3-5):345--351, 2007.

\bibitem{Ch01}
M.~Chastand.
\newblock Fiber-complemented graphs -- {I:} structure and invariant subgraphs.
\newblock {\em Discret. Math.}, 226(1-3):107--141, 2001.

\bibitem{Ch12}
C.~T. Cheng.
\newblock A poset-based approach to embedding median graphs in hypercubes and
  lattices.
\newblock {\em Order}, 29(1):147--163, 2012.

\bibitem{Ch00}
V.~Chepoi.
\newblock Graphs of some {CAT(0)} complexes.
\newblock {\em Adv. Appl. Math.}, 24(2):125--179, 2000.

\bibitem{ChLaRa18}
V.~Chepoi, A.~Labourel, and S.~Ratel.
\newblock Distance and routing labeling schemes for cube-free median graphs.
\newblock {\em CoRR}, abs/1809.10508, 2018.

\bibitem{ChLaRa19}
V.~Chepoi, A.~Labourel, and S.~Ratel.
\newblock Distance labeling schemes for cube-free median graphs.
\newblock In {\em Proc. of MFCS}, volume 138, pages 15:1--15:14, 2019.

\bibitem{CoLeRiSt09}
T.~H. Cormen, C.~E. Leiserson, R.~L. Rivest, and C.~Stein.
\newblock {\em Introduction to Algorithms, 3rd Edition}.
\newblock {MIT} Press, 2009.

\bibitem{Dj73}
D.~Djokovi{\'c}.
\newblock Distance-preserving subgraphs of hypercubes.
\newblock {\em Journal of Combinatorial Theory, Series B}, 14(3):263--267,
  1973.

\bibitem{GeCoMiNa22}
A.~G{\'{e}}ly, M.~Couceiro, L.~Miclet, and A.~Napoli.
\newblock A study of algorithms relating distributive lattices, median graphs,
  and {Formal Concept Analysis}.
\newblock {\em Int. J. Approx. Reason.}, 142:370--382, 2022.

\bibitem{Gr87}
M.~Gromov.
\newblock {\em Hyperbolic Groups}, pages 75--263.
\newblock Springer New York, 1987.

\bibitem{HaImKl99}
J.~Hagauer, W.~Imrich, and S.~Klav{\v{z}}ar.
\newblock Recognizing median graphs in subquadratic time.
\newblock {\em Theor. Comput. Sci.}, 215(1-2):123--136, 1999.

\bibitem{HaImKl11}
R.~Hammack, W.~Imrich, and S.~Klav{\v{z}}ar.
\newblock {\em Handbook of Product Graphs, Second Edition}.
\newblock CRC Press, Inc., 2011.

\bibitem{Ko09}
M.~Kov{\v{s}}e.
\newblock Complexity of phylogenetic networks: counting cubes in median graphs
  and related problems.
\newblock {\em Analysis of complex networks: From Biology to Linguistics},
  pages 323--350, 2009.

\bibitem{LiTa79}
R.~J. Lipton and R.~E. Tarjan.
\newblock A separator theorem for planar graphs.
\newblock {\em SIAM Journal on Applied Mathematics}, 36(2):177--189, 1979.

\bibitem{MoMuRo98}
F.~R. McMorris, H.~M. Mulder, and F.~S. Roberts.
\newblock The median procedure on median graphs.
\newblock {\em Discret. Appl. Math.}, 84(1-3):165--181, 1998.

\bibitem{MuSc79}
H.~M. Mulder and A.~Schrijver.
\newblock Median graphs and {Helly} hypergraphs.
\newblock {\em Discret. Math.}, 25(1):41--50, 1979.

\bibitem{Mu78}
M.~Mulder.
\newblock The structure of median graphs.
\newblock {\em Discret. Math.}, 24(2):197--204, 1978.

\bibitem{Mu80}
M.~Mulder.
\newblock The interval function of a graph.
\newblock {\em Mathematical Centre Tracts, Mathematisch Centrum, Amsterdam},
  1980.

\bibitem{PaBe12}
D.~H. Parks and R.~G. Beiko.
\newblock {Measuring Community Similarity with Phylogenetic Networks}.
\newblock {\em Molecular Biology and Evolution}, 29(12):3947--3958, 2012.

\bibitem{SaNiWi93}
V.~Sassone, M.~Nielsen, and G.~Winskel.
\newblock A classification of models for concurrency.
\newblock In {\em Proc. of {CONCUR}}, volume 715 of {\em Lecture Notes in
  Computer Science}, pages 82--96, 1993.

\bibitem{Sc78}
T.~J. Schaefer.
\newblock The complexity of satisfiability problems.
\newblock In {\em Procs. of {STOC}}, pages 216--226, 1978.

\bibitem{Ve02}
A.~Vesel.
\newblock Recognizing pseudo-median graphs.
\newblock {\em Discret. Appl. Math.}, 116(3):261--269, 2002.

\bibitem{Wi84}
P.~M. Winkler.
\newblock Isometric embedding in products of complete graphs.
\newblock {\em Discret. Appl. Math.}, 7(2):221--225, 1984.

\bibitem{Zi11}
B.~Zimmermann, A.~Röck, G.~Huber, T.~Krämer, P.~M. Schneider, and W.~Parson.
\newblock Application of a west eurasian-specific filter for quasi-median
  network analysis: Sharpening the blade for {mtDNA} error detection.
\newblock {\em Forensic Science International: Genetics}, 5(2):133--137, 2011.

\end{thebibliography}

\end{document}